\documentclass[11pt]{article}
\usepackage[margin=1in]{geometry}
\usepackage{amsmath,amsthm,amssymb}
\usepackage{graphicx}
\usepackage{algorithmic, algorithm}

\theoremstyle{plain} 
\newtheorem{theorem}{Theorem}[section]

\newtheorem{lemma}[theorem]{Lemma}

\newtheorem{definition}[theorem]{Definition}

\theoremstyle{definition}
\newtheorem{claim}[theorem]{Claim}
\newtheorem{remark}{Remark}

\usepackage[backref=page]{hyperref}
\hypersetup{
  colorlinks=true,
  citecolor=blue
}

\def \pr {\mathrm{Pr}}

\def \b {\mathbf{b}}
\def \X {\mathbf{X}}

\def \v {\mathbf{v}}
\def \B {\mathbf{B}}
\def \C {\mathbf{C}}
\def \D {\mathbf{D}}
\def \O {\mathbf{O}}
\def \o {\mathbf{o}}
\def \w {\mathbf{w}}
\def \P {\mathbb{P}}
\def \a {\mathbf{a}}

\DeclareMathOperator*{\E}{\mathbb{E}}

\title{Tight Bounds for the Price of Anarchy of Simultaneous First
  Price Auctions}

\author{George Christodoulou\thanks{Department of Informatics,
    University of Liverpool, UK.  Email:
    \texttt{G.Christodoulou@liverpool.ac.uk} } \and Annam{\'a}ria
  Kov{\'a}cs\thanks{Department of Informatics, Goethe University,
    Frankfurt M., Germany. Email: \texttt{panni@cs.uni-frankfurt.de}}
  \and Alkmini Sgouritsa\thanks{Department of Informatics, University
    of Liverpool, UK. Email: \texttt{a.sgouritsa@liverpool.ac.uk}} \and
  Bo Tang\thanks{Department of Informatics, University of Liverpool,
    UK. Email: \texttt{Bo.Tang@liverpool.ac.uk}}}

\begin{document}

\maketitle
\begin{abstract}
  We study the Price of Anarchy of simultaneous first-price auctions
  for buyers with submodular and subadditive valuations. The current
  best upper bounds for the Bayesian Price of Anarchy of these
  auctions are $e/(e-1)$~\cite{ST13} and $2$~\cite{FFGL13},
  respectively. % Very recently, Roughgarden \cite{Rou14} showed that these
  % bounds are tight, by presenting a reduction from communication lower
  % bounds to lower bounds for the Price of Anarchy. 
  We provide matching lower bounds for
  both cases {\em even} for the case of full information and for
  mixed Nash equilibria via an explicit construction.%  Unlike the lower
  % bounds in \cite{Rou14}, our bounds hold for exact Nash Equilibria
  % and a more restricted class of valuations. 
%
%   In contrast, we show an interesting separation on the Price of
%   Anarchy between the full information and the Bayesian case for
%   buyers with {\em additive} valuations. In that case, we show that
%   mixed Nash equilibria are {\em efficient} whenever they exist, while
%   it is well known that this is not the case for Bayesian Nash equilibria.

  We present an alternative proof of the upper bound of $e/(e-1)$ for
  first-price auctions with fractionally subadditive valuations which
  reveals the worst-case price distribution, that is used as a
  building block for the matching lower bound construction.  

  We generalize our results to a general class of item bidding
  auctions that we call bid-dependent auctions (including first-price
  auctions and all-pay auctions) where the winner is always the
  highest bidder and each bidder's payment depends only on his own
  bid.

% We bring some novel ideas to the theoretical discussion of upper bounding
%   the Price of Anarchy in Bayesian Auctions settings. % The two main techniques that have been used were
% %   smoothness, and the very recent technique of Feldman et al.
%   % None of
% %   them was proved to give tight results, prior to this work.
%   We suggest an alternative way to bid against price distributions. Using our approach we 
%   were able to re-provide the upper bounds of $e/(e-1)$~\cite{ST13}
%   for XOS bidders. An advantage of our approach, is that it reveals a worst-case price
%   distribution, that is used as a building block for the matching lower bound
%   construction.

  Finally, we apply our techniques to discriminatory price multi-unit
  auctions. We complement the results of~\cite{KMST13} for the case of
  subadditive valuations, by providing a matching lower bound of
  $2$. For the case of submodular valuations, we provide a lower bound
  of 1.109. For the same class of valuations, we were able to
  reproduce the upper bound of $e/(e-1)$ using our non-smooth
  approach.
\end{abstract}

%%% Local Variables: 
%%% mode: latex
%%% TeX-master: "poa"
%%% End: 

\section{Introduction}
\emph{Combinatorial auctions} constitute a fundamental, well-studied resource
allocation problem that involves the interaction of $n$ selfish agents
in competition for $m$ in\-di\-vi\-sible resources/goods. The preferences of
each player for different bundles of the items are expressed via a
valuation set function (one per player). The main challenge is to design a {\em
  (truthful) mechanism} that allocates the items in an \emph{efficient} way
in the equilibrium, i.e., so that it
maximizes the {\em social welfare}, which is the sum of the valuations of
the players for the received bundles. Although it is well-known that this can be achieved
optimally by the VCG mechanism~\cite{Vic61,Cla71,Gro73}, unfortunately
this might take exponential time in $m$ and
$n$~\cite{NR99,NR00} %original reference
(unless P=NP).

In practice, several simple {\em non-truthful} mechanisms are
used. The most notable examples are \emph{generalized second price
  (GSP)} auctions used by AdWords \cite{EOS07,Va07}, simultaneous
ascending price auctions for wireless spectrum allocation
\cite{Mil00}, or independent second price auctions on
eBay. Furthermore, in these auctions the expressive power of the
buyers is heavily restricted by the bidding language, so that they are
not able to represent their complex preferences precisely. In light of
the above, Christodoulou, Kov\'acs and Schapira \cite{CKS08} proposed
the study of simple, non-truthful auctions using the \emph{price of
  anarchy (PoA)}~\cite{KP99} as a measure of inefficiency of such
auctions.\footnote{In this setting, the price of anarchy is defined as
  the worst-case ratio of the optimal social welfare over the social
  welfare obtained in a (Bayesian) Nash equilibrium.}

%\smallskip

\paragraph{Item bidding} 

Of particular interest are the so-called {\em combinatorial auctions with
 item-bidding}, from both practical and theoretical aspects. In such an auction, the auctioneer sells each item
by running {\em simultaneously} $m$ independent single-item
auctions. Depending on the type of single-item auctions used, the two main
variants that have been studied are \emph{simultaneous second-price
  auctions (SPAs)}~\cite{CKS08,BR11,FFGL13} and \emph{simultaneous
  first-price auctions (FPAs)}~\cite{HKMN11,ST13,FFGL13}. In both
cases, the bidders are asked to submit a bid for each item. Then each
item is assigned to the highest bidder. The main difference is that in
the former a winner is charged an amount equal to the second highest bid
while in the latter a winner pays his own bid.

FPAs have been shown to be more efficient than SPAs. For general
valuations, Hassidim et al.~\cite{HKMN11} showed that pure equilibria of FPAs are
efficient whenever they exist, but
mixed, and Bayesian Nash equilibria of FPAs can be highly inefficient in settings with
complementarities.  For two important classes
of valuation functions, namely {\em fractionally subadditive} and {\em
  subadditive}\footnote{Fractionally subadditive valuations are also
  known as XOS valuations. For definitions of these valuation
  functions we refer the reader to Section~\ref{sec:pre}.}, for
mixed and Bayesian Nash equilibria,~\cite{HKMN11,ST13} and
\cite{FFGL13} showed that FPAs have lower (constant) price of anarchy
than the respective bounds obtained for
SPAs~\cite{CKS08,BR11,FFGL13}. The current best upper bounds for the price of anarchy in FPA are $e/(e-1)$ for XOS valuations \cite{ST13}, and $2$ for subadditive
valuations \cite{FFGL13} (proven by different techniques). %were shown ; 

\paragraph{Our Contribution}

Following the work of \cite{HKMN11,FFGL13,ST13}, we study the price of
anarchy of FPAs for games with complete and incomplete
information. Our main concern is the development of tools that provide
{\em tight} bounds for the price of anarchy of these auctions. Our
results complement the current knowledge about simultaneous
first-price auctions. We provide {\em matching} lower bounds to the
upper bounds by Syrgkanis and
Tardos~\cite{ST13} and by Feldman et al.~\cite{FFGL13}, showing that {\em
  even} for the case of full information and mixed Nash equilibria the
PoA is at least $e/(e-1)$ for submodular\footnote{In fact our lower
  bound holds even for the class of OXS valuations that is a strict
  subclass of submodular valuations. We refer the reader to
  Section~\ref{sec:pre} for a definition of OXS valuations and for
  their relation to other valuation classes.} valuations (and
therefore for XOS) and $2$ for subadditive valuations\footnote{
  Independently, and after a preliminary version of our
  work~\cite{CKST13}, Roughgarden~\cite{Rou14} showed a general method to provide
  lower bounds for the Price of Anarchy of auctions. We discuss it and
  compare it to our work in the Related Work paragraph.}.

We present an alternative proof of the upper bound of $e/(e-1)$ for
FPAs with \emph{fractionally subadditive} valuations. This bound was shown
before in \cite{ST13} using a general smoothness framework. Our
approach does not adhere to their framework. A nice thing with our
approach, is that it reveals the worst-case price distribution, that
we then use as a building block for the matching lower bound construction.
An immediate consequence of our results is that the price of anarchy
of these auctions {\em stays the same}, for mixed, correlated,
coarse-correlated, and Bayesian Nash equilibria. Only for pure Nash
equilibria it is equal to 1.  Our findings suggest that smoothness may
provide tight results for certain classes of auctions, using as {\em a
  base class} the class of mixed Nash equilibria, and not that of pure
equilibria. This is in contrast to what is known for routing games,
where the respective base class was the class of pure equilibria.

For buyers with {\em additive} valuations (or for the single item auction), we show that any mixed Nash
equilibrium is {\em efficient} in contrast to Bayesian Nash equilibria 
that were previously known not to be always efficient~\cite{Kri02}. 
This suggests an interesting separation between the full and the incomplete information
cases as opposed to other valuation functions (for example submodular and subadditive) 
and other auction formats such as all-pay auctions due to Baye et al.~\cite{BKV96}. % Additionally,
% in Appendix \ref{sec:single_LB_Bayes} we provide a lower bound of
% 1.06 for those.

Then we generalize our results to a class of item bidding auctions
that we call \emph{bid-dependent} auctions. Intuitively, a single item
auction is bid-dependent if the winner is always the highest bidder,
and a bidder's payment depends only on his own bid. Note that both
winner and losers may have to pay. Apart from the FPA (where the
losers pay 0), another notable item-bidding auction that falls into
this class is the simultaneous all-pay (first-price) auction (APA)
\cite{ST13}, where all bidders (even the losers) are charged their
bids. For subadditive valuations, we show that the PoA of simultaneous
bid-dependent auctions is exactly 2, by showing tight upper and lower
bounds. We show that the upper bound technique due to Feldman~et~al.~\cite{FFGL13}
for FPAs, can be applied to {\em all} mechanisms of
this class. Interestingly, although one might expect that FPAs perform
strictly better than APAs, our results suggest that all simultaneous
bid-dependent auctions perform equally well. We note that our upper
bound for subadditive valuations extends the previously known upper
bound of 2 for APAs that was only known for XOS valuations~\cite{ST13} .

Finally, we apply our techniques on \emph{discriminatory price
 multi-unit auctions} \cite{Kri02}. We complement the results by 
de Keijzer et al.~\cite{KMST13} for the case of subadditive valuations, by providing
a matching lower bound of 2, for the standard bidding format. 
For the case of submodular valuations, we
were able to provide a lower bound of 1.109. We were also able to
reproduce their upper bound of $e/(e-1)$ for submodular bids, using
our non-smooth approach.  Note that the previous lower bound for such
auctions was 1.0004~\cite{KMST13} for Bayesian Nash equilibria. Both of
our lower bounds hold for the case of mixed Nash equilibria.

% \smallskip

\paragraph{Related Work}
A long line of research aims to design simple auctions with good
performance guarantee (see e.g. \cite{HR09,CHMS10}). The
(in)efficiency of first-price price auctions has been observed in
economics (cf. \cite{Kri02}) starting from the seminal work by
Vickrey~\cite{Vic61}. Bikhchandani~\cite{Bik99} was the first who studied 
the simultaneous
sealed bid auctions in full information settings and observed the 
inefficiency of their equilibria. 

Christodoulou, Kov\'acs and Schapira \cite{CKS08} extended the concept of PoA to the Bayesian setting and
applied it to item-bidding auctions. Bikhchandani~\cite{Bik99} and then Hassidim et al.~\cite{HKMN11} showed that in
case of general valuations, in FPAs pure Nash equilibria are always
efficient (whenever they exist), whereas for SPAs Fu, Kleinberg and Lavi~\cite{FKL12} proved that the PoA is at
most $2$. For Bayesian Nash equilibria, Syrgkanis and
Tardos~\cite{ST13} and Feldman et al.~\cite{FFGL13} showed improved upper bounds on the \emph{Bayesian
  price of anarchy (BPoA)} for FPAs. Syrgkanis and
Tardos came up with
a general composability framework of smooth mechanisms, that proved to
be quite useful, as it led to upper bounds for several settings, such 
as first price auctions, all-pay auctions and multi-unit auctions.

Only a few lower-bound results are known for the PoA of simultaneous
auctions. For valuations that include {\em complementarities},
Hassidim et al.~\cite{HKMN11} presented an example with PoA $=\Omega(\sqrt{m})$ for
FPA; as suggested in \cite{FFGL13}, similar lower bound can be derived
for SPAs, as well. Under the non-overbidding assumption, Bhawalkar and
Roughgarden~\cite{BR11}
gave a lower bound of $2.013$ for SPAs with subadditive bidders and
$\Omega(n^{1/4})$ for correlated bidders. In \cite{FFGL13}, similar
results are shown under the weak non-overbidding assumption. We
summarize the PoA results for FPAs in table
\ref{tbl:relatedWorkFPA-SPA}.

\begin{table}[t]
\centering
\caption{Summary of the bounds on the PoA of FPAs.
\label{tbl:relatedWorkFPA-SPA}}
{\begin{tabular}{| l | c r | c r |}
\hline
\rule{0pt}{2ex} Valuations & \multicolumn{2}{|c|}{Lower Bound} & \multicolumn{2}{|c|}{Upper Bound}  \\ \hline \hline
\rule{-1mm}{2ex} General, Pure & $1$ &  & $1$ & {\cite{Bik99}} \\
& &  & & {\cite{HKMN11}} \\
General, M-B & $\sqrt{m}$ & {\cite{HKMN11}} & $m$ & {\cite{HKMN11}}\\
SA, M-B & $2$ & {[{\color{red} This paper}]} & $2 $ & {\cite{FFGL13}}\\
XOS, M-B & $e/(e-1)$ & {[{\color{red} This paper}]} & $e/(e-1)$ & {\cite{ST13}} \\
SM, M-B & $e/(e-1)$ & {[{\color{red} This paper}]} & $e/(e-1)$ & {\cite{ST13}} \\
OXS, M-B & $e/(e-1)$ & {[{\color{red} This paper}]} & $e/(e-1)$ & {\cite{ST13}} \\\hline
%& $[2e/(2e-1),e/(e-1)]$ & \emph{\cite{ST13,Rou14}}& $e/(e-1)$ & \emph{\cite{ST13,Rou14}}& $[1,2]$ & \emph{\cite{FFGL13,Rou14}}\\\hline
\end{tabular}}
%
%{\begin{tabular}{| l | c r | c r |}
%\hline
%Valuations & \multicolumn{2}{|c|}{FPA} & \multicolumn{2}{|c|}{SPA}  \\ \hline \hline
%General, Pure & $1$ & \emph{\cite{Bik99}} & $[1,2]$ & \emph{\cite{FKL12}} \\
%& & \emph{\cite{HKMN11}} & & \\
%General, M-B & $[\sqrt{m},m]$ & \emph{\cite{HKMN11}} & $[\sqrt{m},m]$ & \emph{\cite{FFGL13}}\\
%SA, M-B & $[1,2]$ & \emph{\cite{FFGL13}} & $(2 , 4]$ & \emph{\cite{FFGL13}}\\
%XOS, M-B & $[1,e/(e-1)]$ & \emph{\cite{ST13}} & $[1,2]$ & \emph{\cite{CKS08}} \\\hline
%SM, M-B & $[1,e/(e-1)]$ & \emph{\cite{ST13}} & $[1,2]$ & \emph{\cite{CKS08}} \\\hline
%OXS, M-B & $[1,e/(e-1)]$ & \emph{\cite{ST13}} & $[1,2]$ & \emph{\cite{CKS08}} \\\hline
%%& $[2e/(2e-1),e/(e-1)]$ & \emph{\cite{ST13,Rou14}}& $e/(e-1)$ & \emph{\cite{ST13,Rou14}}& $[1,2]$ & \emph{\cite{FFGL13,Rou14}}\\\hline
%\end{tabular}}
%

In the first column, the first argument refers to the valuation class 
and the second argument to the related equilibrium concept. SA and SM stand for subadditive and submodular 
valuations, respectively, and where `M-B' appears the bounds hold for mixed, correlated, coarse 
correlated or Bayesian Nash equilibria.
\end{table}

%\begin{table}[h]
%\centering
%\tbl{Best bounds on the PoA of FPA, known so far.
%\label{tbl:relatedWorkFPA}}
%{\begin{tabular}{| l | c r|}
%\hline
 %& \multicolumn{2}{|c|}{FPA}   \\ \hline
%General, Pure & $1$ & \emph{\cite{HKMN11}} \\
%General, Mixed-Bayes & $[\sqrt{m},m]$ & \emph{\cite{HKMN11}} \\
%Subadditive, Mixed-Bayes & $[1,2]$ & \emph{\cite{FFGL13}} \\
%XOS, Mixed-Bayes & $[1,e/(e-1)]$ & \emph{\cite{ST13}} \\\hline
%%& $[2e/(2e-1),e/(e-1)]$ & \emph{\cite{ST13,Rou14}}& $e/(e-1)$ & \emph{\cite{ST13,Rou14}}& $2$ & \emph{\cite{FFGL13,Rou14}}\\\hline
%\end{tabular}}
%\end{table}
%
%
%\begin{table}[h]
%\centering
%\tbl{Best bounds on the PoA of SPA, known so far.
%\label{tbl:relatedWorkSPA}}
%{\begin{tabular}{| l | c r|}
%\hline
 %&  \multicolumn{2}{|c|}{SPA}  \\ \hline
%General, Pure  & $2$ & \emph{\cite{FKL12}}\\
%General, Mixed-Bayes  & $[\sqrt{m},m]$ & \emph{\cite{FFGL13}}\\
%Subadditive, Mixed-Bayes  & $(2 , 4]$ & \emph{\cite{FFGL13}}\\
%XOS, Mixed-Bayes & $2$ & \emph{\cite{CKS08}} \\\hline
%%& $[2e/(2e-1),e/(e-1)]$ & \emph{\cite{ST13,Rou14}}& $e/(e-1)$ & \emph{\cite{ST13,Rou14}}& $2$ & \emph{\cite{FFGL13,Rou14}}\\\hline
%\end{tabular}}
%\end{table}

Very recently and independently, Roughgarden~\cite{Rou14} presented a very
elegant methodology to provide PoA lower bounds via a reduction from
communication or computational complexity lower bounds for the
underlying optimization problem. One consequence of his reduction is a
general lower bound of $2$ and $e/(e-1)$ for the PoA of {\em any}
simple auction (including item-bidding auctions) with subadditive and
fractionally subadditive bidders, respectively. Therefore, there is an
overlap with our results for these two classes of valuations.
We show these lower bounds via an explicit construction for FPAs (and
also for bid-dependent auctions). 

We emphasize that these two approaches are incomparable in the
following sense.  On the one hand, the results in \cite{Rou14} hold
for more general formats of combinatorial auctions than the ones we
study here.  On the other hand, our $e/(e-1)$ lower bound holds even
for more special valuation functions where \cite{Rou14}'s results are
either weaker ($2e/(2e-1)$ for submodular valuations) or not
applicable.  For the case of submodular valuations, Feige and
Vondr{\'{a}}k~\cite{FV10}
showed that a strictly higher than $1-1/e$ amount of the optimum
social welfare can be obtained in polynomial communication and for
gross substitute valuations (and therefore for its subclass, OXS
valuations), Nisan and Segal~\cite{NS06} showed that exact efficiency can be obtained
in polynomial communication. These two results show that
\cite{Rou14}'s technique does not provide tight lower bounds for the
settings studied in this paper. We also note that the PoA lower bound
obtained by \cite{Rou14}'s reduction can only be applied to
approximate Nash equilibria while our results apply to exact Nash
equilibria. Further, our PoA lower bound proof for subadditive
valuations uses a simpler construction than the proof in \cite{Rou14}
and it holds even for the case of only $2$ bidders and identical items
(multi-unit auction).  Finally, it should be stressed that none of our
lower bounds for multi-unit auctions can be derived from
\cite{Rou14}.%, that is another advantage of our constructions.

% This is a different
% approach from that of \cite{Rou14}, and our $e/(e-1)$ lower bound
% holds even for a subclass of submodular valuations, the \emph{OXS}
% valuations, whereas the best lower bound that can be derived from the
% reduction in \cite{Rou14} for (even) submodular valuations, is
% $2e/(2e-1)$. We note also that the PoA lower bound obtained by
% Roughgarden's reduction can only be applied to approximate Nash
% equilibria while our results apply to exact Nash equilibria.

Markakis and Telelis~\cite{MT12} studied uniform price \emph{multi-unit}
auctions. De Keijzer et al.~\cite{KMST13} bounded the BPoA for several formats of
multi-unit auctions with first or second price rules. Auctions
employing greedy algorithms were studied by Lucier and Borodin~\cite{LB10}. A number of
works \cite{PT10, CKK11, Rou12} studied the PoA of generalized SPAs in
the full information and Bayesian settings and even with correlated
bidders \cite{LP11}. Chawla and Hartline~\cite{CH13} proved that for the generalized FPAs
with symmetric bidders, the pure Bayesian Nash equilibria are unique
and always efficient.

\paragraph{Organization of the paper}
We introduce the necessary background and notation in Section
\ref{sec:pre}. Then, we present the tight bounds for FPAs with
fractionally subadditive, subadditive  and additive valuations
in Sections \ref{sec:upper}, \ref{sec:subadditive_LB} and
\ref{sec:additive}, respectively. In Section \ref{sec:general}, we show
how our results can be generalized to a class of auctions that we call
bid-dependent. Finally, we apply our techniques to get
bounds on the PoA of discriminatory multi-unit auctions in Section
\ref{sec:discriminatory}.

%%% Local Variables: 
%%% mode: latex
%%% TeX-master: "poa"
%%% End: 

\section{Preliminaries}
\label{sec:pre}

\emph{Simultaneous first-price auctions} constitute a simple type of
combinatorial auctions. In a combinatorial auction with $n$
\emph{players} (or \emph{bidders}) and $m$ \emph{items}, every player 
$i\in [n]$ has a valuation for each subset of items, given by a
valuation function $v_i: 2^{[m]} \rightarrow \mathbb R_{\geq 0},$
where $v_i\in V_i$ for some possible set of valuations $V_i.$ A
valuation profile for all players is $\v=(v_1,v_2,\ldots, v_n)\in \times_i V_i.$ The $v_i$
functions are monotone and normalized, that is, $S\subseteq
T\Rightarrow v_i(S)\leq v_i(T),$ and $v_i(\emptyset)=0.$ We use the short notation $v_i(j)=v_i(\{j\}).$ 

In the \emph{Bayesian} setting, the valuation of each player $i$ is
drawn from $V_i$ according to some known distribution $D_i.$ We assume
that the $D_i$ are independent (and possibly different) over the
players. In the \emph{full information} setting the valuation $v_i$ is
fixed and known by all other players for all $i\in [n].$ Note that the
latter is a special Bayesian combinatorial auction, in which player
$i$ has valuation $v_i$ with probability $1.$

An \emph{allocation} $\X=(X_1,\ldots,X_n)$  is a partition  of the items 
(allowing empty sets $X_i$), so that each item is assigned to exactly one player. 
The most common global objective in combinatorial auctions is to maximize
the sum of the valuations of the players for their received sets of
items, i.e., to maximize the \emph{social welfare} $SW(\X)$ of the
allocation, where $SW(\X)=\sum_{i\in [n]} v_i(X_i).$ Therefore, for an
\emph{optimal allocation} $\O(\v)=\O=(O_1,\ldots,O_n)$ the
value $SW(\O)$ is maximum among all possible allocations.

In a \emph{simultaneous (or item bidding)} auction every player $i\in
[n]$ submits a non-negative bid $b_{ij}$ for each item $j\in [m].$ The
items are then allocated by independent auctions: for each $j\in [m],$
the bidder $i$ with the highest bid $b_{ij}$ receives the item.  We
consider the case when the payment for each item is the \emph{first
  price} payment: a player pays his own bid (the highest bid) for
every item he receives.  Our (upper bound) results hold for arbitrary
randomized tie-breaking rules. %as long as it chooses among the valid allocations.
Note that with such a rule, for any fixed $\b=(b_1,\ldots,b_n),$ 
the probabilities for the players to get a particular item are fixed. 

%  Switching to
% more precise game theoretic definitions, we define a \emph{pure
%   (bidding) strategy $b_i$} for player $i$ to be a vector of bids for
% the $m$ items $b_i=(b_i(1),\ldots ,b_i(m)).$ As usual, $\b_{-i}$
% denotes the strategies of all players except for $i.$ The pure
% strategy profile of all bidders is then $\b=(b_i,\b_{-i}).$ For any
% set $S$ of items, we use the notation $b_i(S)=\sum_{j\in S}b_i(j)$ to
% denote the sum of the bids of player $i$ for $S$. A \emph{mixed
%   strategy $B_i$} of player $i$ is a probability distribution over
% pure strategies. $\B=(B_1,\ldots,B_n)$ is a profile of mixed
% strategies.

For a given bid vector $b_i,$ item $j\in [m]$ and a subset of  items $S\subseteq [m]$ 
we use the notation $b_i(S)=\sum_{j\in S} b_{ij},$ and $b_i(j)=b_{ij}.$
Assume that the players submitted bids for the items according to
$b_i=(b_{i1},\ldots ,b_{im})$ and the simultaneous first-price
auction yields the allocation $\X(\b)$. For simplicity, we use
$v_i(\b)$ and $SW(\b)$ instead of $v_i(X_i(\b))$ and $SW(\X(\b)),$ to
express the valuation of player $i$ and the social welfare for the
allocation $\X(\b)$ if $\X$ is clear from the context. The
\emph{utility $u_i$} of player $i$ is defined as his valuation for the
received set, minus his payments: $u_i(\b)= v_i(X_i(\b))-b_i(X_i(\b)).$

\subsection{Bidding strategies, Nash equilibria, and the price of anarchy}

A \emph{pure (bidding) strategy $b_i$} for player $i$ is a vector of
bids for the $m$ items $b_i=(b_{i1},\ldots ,b_{im}).$ As usual,
$\b_{-i}$ denotes the strategies of all players except for $i.$ The
pure strategy profile of all bidders is then $\b=(b_1,b_2,\ldots,
b_n).$ 
% For any set $S$ of items, we use the notation
% $b_i(S)=\sum_{j\in S}b_{ij}$ for the sum of the bids of player $i$ for
% $S$. We also use $b_i(j)$ to denote $b_{ij}$ if $S$ is a singleton.

A \emph{mixed strategy $B_i$} of player $i$ is a probability
distribution over pure strategies. Let $\B=(B_1,\ldots,B_n)$ be a
profile of mixed strategies. Given a profile $\B$, we fix the notation
for the following \emph{cumulative distribution functions (CDF):} unless defined otherwise,
$G_{ij}$ is the CDF of the bid of player $i$ for item $j;$ $F_{j}$
is the CDF of the highest bid for item $j$ in $\b,$ and $F_{ij}$ is
the CDF of the highest bid for item $j$ in $\b_{-i}.$  Observe that
$F_j=\Pi_k G_{kj},$ and $F_{ij}=\Pi_{k\neq i} G_{kj}.$ We also use 
$\varphi_{ij}(x)$ to denote the probability that player $i$ gets item $j$
by bidding $x.$ Then $\varphi_{ij}(x)\leq F_{ij}(x)$ due to a possible tie in $x.$

% For simplicity, we use $v_i(\b)$ and $SW(\b)$ instead of
% $v_i(X_i(\b))$ and $SW(\X(\b)),$ to express the social welfare and the
% valuation of player $i$ for the allocation $\X(\b)$ if $\X$ is clear
% from the context.  The \emph{utility $u_i$} of player $i$ is defined
% as the difference between her valuation for the received allocation
% $v_i(\b)$ and her payments $q_i(\b)$: $u_i(\b)=v_i(\b)-q_i(\b).$

We review five standard equilibrium concepts studied in this paper:
pure, mixed, correlated, coarse correlated and Bayesian Nash
equilibria. The first four of them are for the \emph{full information}
setting and the last one is defined in the \emph{Bayesian}
setting. Let $\v=(v_1,\ldots,v_n)$ be the players valuation
functions. In the Bayesian setting, $v_i$ is drawn from $V_i$
according to some known distribution. Let $\B$ denote a distribution over bidding profiles $\b$ of the players.
%Recall that $\B=(B_1,\ldots,B_n)$ denotes the bidding strategies of the players.  
Then, $\B$ is called a

\smallskip

\noindent-- \emph{pure Nash equilibrium}, if %$B_i$ is a pure strategy $b_i$ 
$\B$ is a pure strategy profile $\b$ and $u_i(\b)\ge u_i(b'_i, \b_{-i})$.

\noindent-- \emph{mixed Nash equilibrium}, if $\B=\times_i B_i$ and
$\E_{\b\sim\B}[u_i(\b)]\ge \E_{\b_{-i}\sim
  \B_{-i}}[u_i(b'_i,\b_{-i})]$.

\noindent-- \emph{correlated Nash equilibrium}, if
$\E_{\b\sim\B}[u_i(\b)|b_i]\ge \E_{\b\sim\B}[u_i(b'_i,\b_{-i})|b_i]$.

\noindent-- \emph{coarse correlated Nash equilibrium}, if
$\E_{\b\sim\B}[u_i(\b)]\ge \E_{\b\sim \B}[u_i(b'_i,\b_{-i})]$.

\noindent-- \emph{Bayesian Nash equilibrium}, if
$B_i(\v)=\times_iB_i(v_i)$ and $\E_{\v_{-i},\b\sim \B(\v)}[u_i(\b)]\ge
\E_{\v_{-i},\b_{-i}\sim \B_{-i}(\v_{-i})}[u_i(b'_i, \b_{-i})].$

\smallskip

\noindent where the given inequalities hold for all players $i$ and
(pure) deviating bids $b'_i$. It is well known that each one of the
first four classes is contained in the next class, i.e., pure
$\subseteq$ mixed $\subseteq$ correlated $\subseteq$ coarse
correlated. If we regard the \emph{full information} setting as a
special case of the \emph{Bayesian} setting, we also have pure $\subseteq$
mixed $\subseteq$ Bayesian.

For a given auction and fixed valuations $\v$ of the bidders, let $\O$
be an optimal allocation. Then for this auction (game) the \emph{price
  of anarchy in pure equilibria} is $\mathrm{PoA}= \max_{\b\,
  \mathrm{pure} \,\mathrm{Nash}} \frac{SW(\O)}{SW(\b)}; $ Given a
class of auctions, the \emph {price of anarchy (PoA)} for this type of
auction is the worst case of the above ratio, over
all auctions of the class, valuation profiles $\v$ and bidding
profile $B$. For the other four types of equilibria, the price of
anarchy can be defined analogously.

For the expected utility of a
given bidder $i$ we often use the short notation $\E[u_i]$ (if $\B$ is clear from the context) or $u_i(\B)$ to denote $\E_{\b\sim
  \B}[u_i(\b)]$.  Similarly, for fixed $b'_i,$ we use
$\E[u_i(b'_i)]=\E_{\b_{-i}\sim \B_{-i}}[u_i(b'_i,\b_{-i})].$ We also
use $\E_{\v}$ instead of $\E_{\v\sim \mathbf{D}}.$

\subsection{Types of valuations}
Our results concern different classes of valuation functions, which we
define next, in increasing order of inclusion.  Let $v: 2^{[m]}
\rightarrow \mathbb R_{\geq 0},$ be a valuation function. Then $v$ is called

\smallskip

\noindent-- \emph{additive}, if $v(S)=\sum_{j\in S} v({j});$

\noindent-- \emph{multi-unit-demand} or \emph{OXS}, if for some $k$ there exist 
$k$ unit demand valuations $v^r$, $r \in [k]$ (defined as 
$v^r(S)=\max_{j\in S}v^r(j)$), such that
$v(S)=\max_{S=\dot\bigcup_{r\in[k]} S_r }\sum_{r\in [k]} v^r(S_r);$\footnote{{\small $\dot\bigcup$} stands for disjoint union}

\noindent-- \emph{submodular}, if $v(S\cup T)+v(S\cap T)\leq
v(S)+v(T);$

\noindent-- \emph{fractionally subadditive} or \emph{XOS}, if $v$ is
determined by a finite set of \emph{additive} valuations $f_\gamma$
for $\gamma\in \Gamma,$ so that $v(S)= \max_{\gamma\in \Gamma}
f_\gamma(S);$

\noindent-- \emph{subadditive}, if $v(S\cup T)\leq v(S)+v(T);$

\smallskip

\noindent where the given equalities and inequalities must hold for arbitrary item sets
$S,T\subseteq [m].$  It is well-known that each one of
the above classes is strictly contained in the next class, e.g., an
additive set function is always submodular but not vice versa, a
submodular is always XOS, etc. \cite{Fei06}.  As an equivalent
definition, submodular valuations are exactly the valuations with
\emph{decreasing marginal values}, meaning that $v(\{j\}\cup T)-v(T)\leq
v(\{j\}\cup S)-v(S)$ holds for any item $j,$ given that $S\subseteq T.$

\section{Submodular valuations}
\label{sec:upper}

In this section we present a lower bound of ${\frac{e}{e-1}}$ for the mixed PoA in
simultaneous first price auctions with OXS and therefore, submodular and fractionally subadditive valuations. 
This is a matching lower bound to the results by Syrgkanis and
Tardos~\cite{ST13}.

In order to the explain the key properties of the instance proving a tight lower bound, first
we discuss a new approach to obtain the same upper bound for the PoA of a first price single item auction
as in \cite{ST13}. While the upper
bound that we derive with the help of this idea, can also be
obtained based on the very general \emph{smoothness} framework
\cite{Rou12,ST13,KMST13}, the approach we introduce here does not
adhere to this framework.\footnote{Roughly, because the pure deviating
  bid $a$ that we identify, depends on the other players' bids
  $\b_{-i}$ in the Nash equilibrium.} The strength of our approach consists in its
potential to lead to better (in this case tight) \emph{lower} bounds,
as we demonstrate subsequently.

%\iffalse
%In order to keep the presentation pure, and to focus on the main
%ingredients, in this section we illustrate this approach on a single
%item auction with full information. For this case\footnote{Note that this bound is not tight, see
%Section~\ref{sec:additive}.} we obtain the bound
%of $\frac{e}{e-1}$. We show that the same idea can be
%used for obtaining the same bound for discriminatory auctions, and
%various generalizations of simultaneous auctions with $m$ items: XOS
%valuations in Bayesian NE; and XOS valuations in coarse correlated
%equilibria. It is also trivial (so the proof is omitted) to extend our
%approach of single item to arbitrary valuations (with coarse
%correlated equilibria) when the players receive at most one item in
%the optimal allocation.
%\fi

% These generalizations, demonstrating that our approach is rather
% powerful, can be found in Appendix~\ref{sec:upper}.

\subsection{PoA Upper Bound for Single Item Auctions}
\label{sec:one_item_upper}

\begin{theorem}
  The PoA of mixed Nash equilibria in first-price single-item auctions
  is at most $\frac{e}{e-1}.$
\end{theorem}

\begin{proof}
  Let $\v=(v_1,\ldots,v_n)$ be the valuations of the players, and
  suppose that $v_i=\max_{k\in [n]} v_k.$ We fix a mixed Nash
  equilibrium $\B=(B_1,B_2,\ldots,B_n).$ Let $p_i$ denote the highest
  bid in $\b_{-i},$ and $F(x)=F_i(x)$ be the cumulative distribution
  function (CDF) of $p_i,$ that is, $F(x)=\P_{\b_{-i}\sim
    \B_{-i}}[p_i\leq x].$
The following lemma prepares the ground for the selection of an
appropriate deviating bid.
  \begin{lemma}\label{lem:tie}
    For any pure strategy $a$ of player $i$,
    $\E_{\b\sim\B}[u_i(\b)]\geq F(a)(v_i-a)$.
  \end{lemma}
  \begin{proof}
    If $F$ is continuous in $a,$ then $F(a)=\P[p_i\leq
    a]=\P[p_i<a],$ tie-breaking in $a$ does not matter, and $F(a)$
    equals also the probability that bidder $i$ gets the item if he
    bids $a.$ Therefore,
    $F(a)(v_i-a)=\E_{\b_{-i}\sim\B_{-i}}[u_i(a,\b_{-i})]\leq
    \E_{\b\sim\B}[u_i(\b)],$ since $\B$ is a Nash equilibrium.
    If $F$ is not continuous in $a$ ($\P[p_i=a]>0$), then, as a
    CDF, it is at least right-continuous. By the previous argument
    $\E[u_i(\b)]\geq F(a+\epsilon)(v_i-a-\epsilon)$ holds for every
    $x=a+\epsilon$ where $F$ is continuous, and the lemma follows by
    taking $\epsilon\rightarrow 0.$
  \end{proof}

  Since in a Nash equilibrium the expected utility of every (other) player is
  non-negative, by summing over all players, it holds that $\sum_{k=1}^n \E[u_k(\b)]\geq
  F(a)(v_i-a).$   
  
  On the other hand, for any fixed bidding profile $\b$ we have $u_k(\b)=v_k(X_k(\b))-b_k(X_k(\b)),$ where $b_k(X_k(\b))=b_k$
  whenever $b_k$ is a winning bid, and $b_k(X_k(\b))=0$ otherwise.
  Let  $b_{\max}$ be the maximum bid in $\b.$ 
  By taking expectations with regard to $\b\sim \B,$ and summing over the players, 
  $\E[\sum_k  u_k(\b)]=\E[\sum_k (v_k(X_k(\b))-b_k(X_k(\b)))] 
  =\E[\sum_k v_k(X_k(\b))- b_{\max}]
  =\E[SW(\b)]-\E[b_{\max}].$
By combining this with Lemma~\ref{lem:tie},
  we obtain
  \begin{equation}
    \label{eq:lambda0}\E[SW(\b)]=\E\left[\sum_{k=1}^n
      v_k(X_k(\b))\right]\geq F(a)(v_i-a)+\E[b_{\max}]\geq
    F(a)(v_i-a)+\E[p_i], \end{equation}
%   We found that $\E[SW(\b)]\geq F(a)(v_i-a)+\E[p_i]$ holds for an
%   arbitrary (deviating) bid $a.$ 
 for any (deviating) bid $a.$ (Analogues of this derivation are standard
  in the simultaneous auctions literature.)
We choose the bid $a^*$ that {\em maximizes} the right hand side of
(\ref{eq:lambda0}), i.e. $a^*=\arg\max_{a}F(a)(v-a)$ (see Figure~\ref{fig:F_and_tightF} (a) for an illustration). Then, in order to upper bound the PoA, we look for the {\em maximum} value of $\lambda$, such that,

%Our objective is to find the maximum
%  value of $\lambda$, such that for \emph{some} pure strategy $a$,

  \begin{equation}
    \label{eq:lambda1}
    F(a^*)(v_i-a^*) + \E[p _i] \geq \lambda v_i.
  \end{equation}
%Figure~\ref{fig:F} shows schematically the relation of  $F(a^*)(v-a^*)$ and $\E[p]$.
%   \begin{lemma}\label{lem:1.58upper1}
%     For any non-negative random variable $p$ with CDF $F(x),$ and any
%     fixed number $v$, there exists a real number $a\leq v$ such that
%     $F(a)(v-a) + \E[p] \geq \left(1-\frac{1}{e}\right)v$.
%   \end{lemma}
\iffalse
Next we argue that $a^*$ leads to tight bounds against both {\em pure} and {\em
  mixed} strategies. The following lemma serves the latter purpose\footnote{Note that if $F$ is pure, then it easy to
  verify that the $\lambda$ that corresponds to $a^*$ is equal to 1,
  leading to a PoA=1.}.
\fi
The following lemma settles the maximum value of such $\lambda$ as $\,1-\frac{1}{e}\,$ for mixed equilibria.\footnote{If $\B$ is a pure equilibrium, then it is easy to verify that $F$ is  a step function, furthermore $a^*=p_i,$ and inequality~(\ref{eq:lambda1}) boils down to $1\cdot (v_i-a^*)+a^*=1\cdot v_i$} This will complete the proof of the theorem, since  by (\ref{eq:lambda0}) and $SW(\O)=v_i$ we obtain $\E[SW(\b)]\geq (1-\frac{1}{e})\cdot SW(\O).$ 
  
\begin{figure}[t]
	\centering
\begin{tabular}{c c}
\includegraphics[scale=0.45]{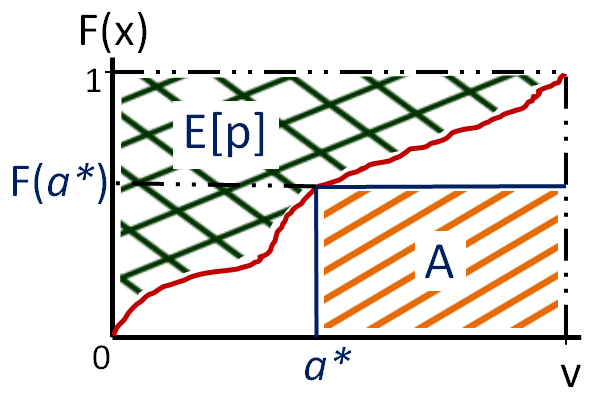} &
\includegraphics[scale=0.35]{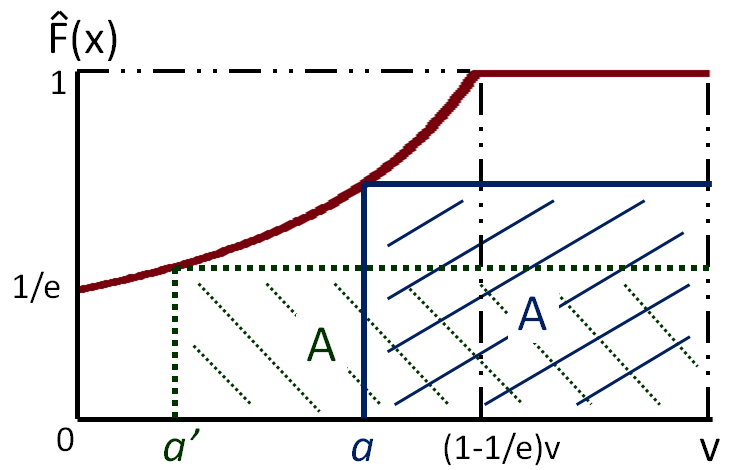}\\
$(a)$&$(b)$
\end{tabular}
 \caption{Figure $(a)$ is a schematic illustration of the
        expression $F(a^*)(v-a^*) + \E[p]$, where $A=F(a^*)(v-a^*)$. Figure $(b)$ shows the CDF $\hat{F}(x)$,which makes all the inequalities of Lemma
        \ref{lem:1.58upper1} tight, i.e. for every $x \in
        [0,\left(1-\frac{1}{e}\right)v]$, $F(x)(v-x)=A=\frac{v}{e}$.}
    \label{fig:F_and_tightF}
\end{figure}

  \begin{lemma}\label{lem:1.58upper1}
    For any non-negative random variable $p$ with CDF $F,$ and any
    fixed value $v$, it is true that 
    $$F(a^*)(v-a^*) + \E[p] \geq \left(1-\frac{1}{e}\right)v.$$
  \end{lemma}
  \begin{proof}
  \iffalse
%     \begin{figure}[htmp]
%       \centering

%       \includegraphics[scale=0.4]{1st_1item_upper_bound.jpg}

%       \caption{This figure gives a schematic illustration of the
%         expression $F(a)(v-a) + \E[p]$. $A=F(a)(v-a)$ and
%         $B=\E[p]$. This specific CDF is the $\hat{F}(x)$, which
%         tightens all the inequalities of Lemma
%         \ref{lem:1.58upper1}. So, for every $x \in
%         [0,\left(1-\frac{1}{e}\right)v]$, $F(x)(v-x)=A$.}

%       \label{fig:tightF}

%     \end{figure}

    % We set $a$ to be the value in $[0,v]$ that maximizes
%     $F(x)(v-x)$.
\fi
    Set $A= F(a^*)(v-a^*), $ for $a^*=\arg\max_{a}F(a)(v-a).$  We use the fact that the expectation of a non-negative random variable with CDF $F$ can be calculated as $\E[x]=\int_0^\infty (1-F(x))dx.$
    \iffalse
%\geq F(x)(v-x)\quad \forall
 %   x\in[0,v]
 \fi
     Thus, %We also use that $\E[p] = \int_0^{\infty}(1-F(x))dx$.
    \begin{eqnarray*}
      F(a^*)(v-a^*) + \E[p] %& = A + \int_0^{\infty}(1-F(x))dx\\
      &\geq& A + \int_0^{v-A}(1-F(x))dx= A+(v-A)- \int_0^{v-A}F(x)dx\\
      %= v- \int_0^{v-A}F(x)dx \\
			&\geq& v- \int_0^{v-A}\frac{A}{v-x}dx 
      = v+A\ln\left(\frac{A}{v}\right)\\ &\geq&
      v+\frac{v}{e}\ln\left(\frac{1}{e}\right) = \left(1 -
        \frac{1}{e}\right)v,
    \end{eqnarray*}
    where the last inequality is due to the fact that $A\ln(\frac{A}{v})$ is
    minimized for $A=\frac{v}{e}$.
  \end{proof} \end{proof}
  % The theorem follows by applying Lemma \ref{lem:1.58upper1} to
%   inequality \eqref{eq:lambda1}, we get $\lambda \geq 1-\frac{1}{e}$,
%   therefore PoA$\leq\frac{1}{1-1/e} \approx 1.58$. 

%The theorem follows by applying Lemma \ref{lem:1.58upper1} to (\ref{eq:lambda0}).
%   inequality \eqref{eq:lambda1}, we get $\lambda \geq 1-\frac{1}{e}$,
%   therefore PoA$\leq\frac{1}{1-1/e} \approx 1.58$. 

\paragraph{Worst-case price distribution}

The CDF $F(x)$ that makes all the inequalities of (the proof of) Lemma
\ref{lem:1.58upper1} tight (see Figure~\ref{fig:F_and_tightF}(b)), is
\[
\hat{F}(x) = \left\{
  \begin{array}{l l}
    \frac{v}{e(v-x)} & \quad \text{, for } x\leq \left(1-\frac{1}{e}\right)v\\
    1 & \quad \text{, for } x > \left(1-\frac{1}{e}\right)v\\
  \end{array} \right.
\]

\iffalse
 \begin{figure}[htmp]
      \centering
      \includegraphics[scale=0.4]{tight_upper_bound.png}
      \caption{The CDF $\hat{F}(x)$ makes all the inequalities of Lemma
        \ref{lem:1.58upper1} tight, i.e. for every $x \in
        [0,\left(1-\frac{1}{e}\right)v]$, $F(x)(v-x)=A=\frac{v}{e}$.}
      \label{fig:tightF}
    \end{figure}
\fi

Observe that for $x\leq \left(1-\frac{1}{e}\right)v$, $\hat{F}(x)(v-x)
= \frac{v}{e}$ and for $x> \left(1-\frac{1}{e}\right)v$,
$\hat{F}(x)(v-x) = v-x < v-(1-\frac{1}{e})v=\frac{v}{e}$.  So, the bid
that maximizes the quantity $\hat{F}(a)(v-a)$ is any value $a \in [0,
\left(1-\frac{1}{e}\right)v]$. %For any such value $A = \frac{v}{e}$.
The given distribution $\hat F$ for $p_i$ makes inequality
\eqref{eq:lambda1} tight. In order to construct a (tight) lower bound for the PoA,
we also need to tighten the inequalities in \eqref{eq:lambda0}. Note
that the inequality of Lemma~\ref{lem:tie} is tight for {\em all} $a \in [0,
\left(1-\frac{1}{e}\right)v].$ Intuitively, we need to construct a Nash
equilibrium, where the CDF of $p_i$ is equal to $\hat{F}(x)$ and $b_i$
doesn't exceed $p_i.$ We present a construction (with many items) in Section~\ref{sec:submodular_LB}.

% \begin{remark}
% There is an interesting comparison between our approach and the smoothness. 

% technique, used by Syrgkanis and Tardos~\cite{ST13} that achieves the same upper
% bound for single item auctions. % (the bound itself is due to Feige
% % \cite{Fei06})
%  The authors of \cite{ST13} found a mixed strategy $a$
% independent of the $b_{-i}$, that gives the same result for every
% $F(x)$. On the other hand, we found a pure strategy $a$ that depends
% on the $F(x)$, i.e. on the $b_{-i}$.  % We remark, furthermore, that the
% % argument remains valid for coarse correlated equilibria (see, e.g.,
% % \cite{BR11}).  \smallskip
% \end{remark}

% \paragraph{Minimax Theorem}

\begin{remark}\label{rem:discussion-smoothness}
Here we discuss our technique and the smoothness technique that achieves
the same upper bound~\cite{ST13}. In \cite{ST13}, a particular mixed
bidding strategy $A_0$ was defined for each player $i$, such that for
every price $p=\max_{i'\neq i}b_{i'}$, $\E_{A_0}[u_i(A_0,p)]+p \ge
v(1-1/e)$. If we denote
$g(A,F)=\E_{A,F}[u_i(a,p)+p]$, it can be deduced that
$\max_A\min_pg(A,p)\ge v(1-1/e)$. In Lemma~\ref{lem:1.58upper1} we
show that $\min_F\max_ag(a,F)\ge v(1-1/e)$. Moreover, we prove that the inequality is
tight by providing the {\em minimizing} distribution $\hat{F}$, such that
$\max_ag(a,\hat{F})=v(1-1/e)$. By the Minimax Theorem, 
$\min_F\max_ag(a,F)=\max_A\min_pg(A,p)=v(1-1/e)$. 
%Equality shows that the upper bound obtained in \cite{ST13} cannot be further improved.
 One advantage of our approach is that it can be coupled with a worst-case
distribution $\hat{F}$ that serves as an optimality certificate of
the method. Moreover, if one can convert $\hat{F}$ to Nash
Equilibrium strategy profile (see Section~\ref{sec:submodular_LB}), a tight Nash equilibrium construction is
obtained; this can be a challenging task, though.   
\end{remark}

\subsection{Tight Lower Bound}
\label{sec:submodular_LB}

Here we present a tight lower bound of ${\frac{e}{e-1}}$ for the mixed PoA in
simultaneous first price auctions with OXS valuations. This
implies a lower bound for submodular and fractionally subadditive valuations. % In
% Appendix~\ref{sec:LBd-dimensions} we show a class of instances with
% inefficient mixed equilibria generalizing the
% construction presented in this section. 

\begin{theorem} The price of anarchy of simultaneous first price
  auctions with full information and OXS valuations is at least
  ${\frac{e}{e-1}}\approx 1.58.$
\end{theorem}

\begin{proof}
  We construct an instance with $n+1$ players and $n^n$ items. % (in
  % Appendix \ref{sec:LBd-dimensions} we present a more general
  % instance with $n^d$ items)
  We define the set of items as $M=[n]^n,$ that is,
  they correspond to all the different vectors $w=(w_1, w_2, ... ,
  w_n)$ with $w_i\in [n]$ (where $[n]$ denotes the set of integers
  $\{1,\ldots,n\}$). Intuitively, they can be thought of as the
  nodes of an $n$ dimensional grid, with coordinates in $[n]$ in each
  dimension.

  We call player $0$ the \emph{dummy} player, and all other players $i\in [n]$ \emph{real} players.
  We associate each \emph{real} player $i$ with one of the
  dimensions (directions) of the grid. In particular, for any fixed
  player $i,$ his valuation for a subset of items $S\subseteq M$ is the
  size (number of elements) in the $n-1$-dimensional projection of $S$
  in direction $i.$ Formally, $$v_i(S)=|\{w_{-i}\,|\,\exists w_i\,\,
  s.t. \,\,(w_i,w_{-i})\in S \}|.$$ It is straightforward to check
  that $v_i$ has decreasing marginal values, and is therefore
  submodular\footnote{These valuations are also OXS. In the definition of OXS valuations (Section \ref{sec:pre}), we set
	$k=n^{n-1}$ and for the unit-demand valuations corresponding to player $i$ the following holds: 
	if item $j$ corresponds to $w=(w_1, w_2,...,w_n)$ then for $r \in [k]$, 
	$v_i^r({j})=1$, if $w_{-i}$ is the $n$-ary representation of 
	$r$ and $v_i^r({j})=0$, otherwise.}. 
	The  \emph{dummy} player $0$ has
  valuation $0$ for any subset of items.

  Given these valuations, we describe a mixed Nash equilibrium
  $\B=(B_1,\ldots ,B_n)$ having a PoA arbitrarily close to $e/(e-1),$
  for large enough $n.$ The dummy player bids $0$ for every item, and
  receives the item if all of the real players bid $0$ for it. The
  utility and welfare of the dummy player is always $0.$ For real
  players the mixed strategy $B_i$ is the following. Every player $i$
  picks a number $\ell\in [n]$ uniformly at random, and an $x$
  according to the distribution with CDF
$$G(x)=(n-1)\left(\frac{1}{\left(1-x\right)^{\frac{1}{n-1}}}-1\right),$$
where $x\in\left[0,1-\left(\frac{n-1}{n}\right)^{n-1}\right]$.
Subsequently, he bids $x$ for every item $w=(\ell, w_{-i})$, with
$w_i=\ell$ as $i^{th}$ coordinate, and bids $0$ for the rest of the
items, see Figure~\ref{fig:exampleLBsub} for the cases of $n=2$ and $n=3$.
That is, in any $b_i$ in the support of $B_i,$ the player bids a
positive $x$ only for an $n-1$ dimensional slice of the items. Observe
that $G(\cdot)$ has no mass points, so tie-breaking matters only in case of
$0$ bids for an item, in which case player $0$ gets the item.

\begin{figure}
\begin{tabular}{c c c}
\includegraphics[scale=0.25]{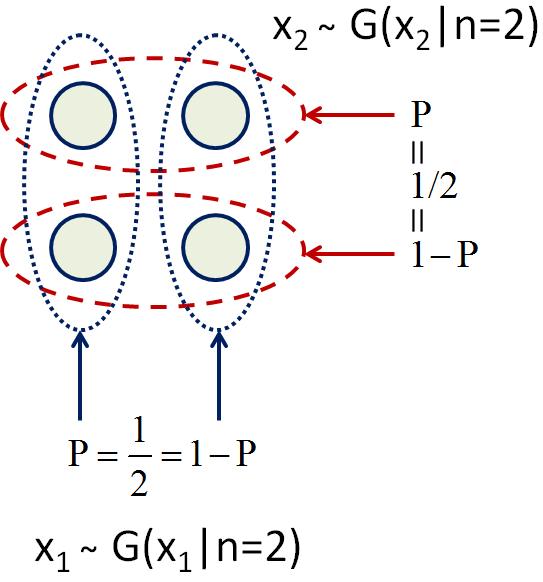} &
\includegraphics[scale=0.25]{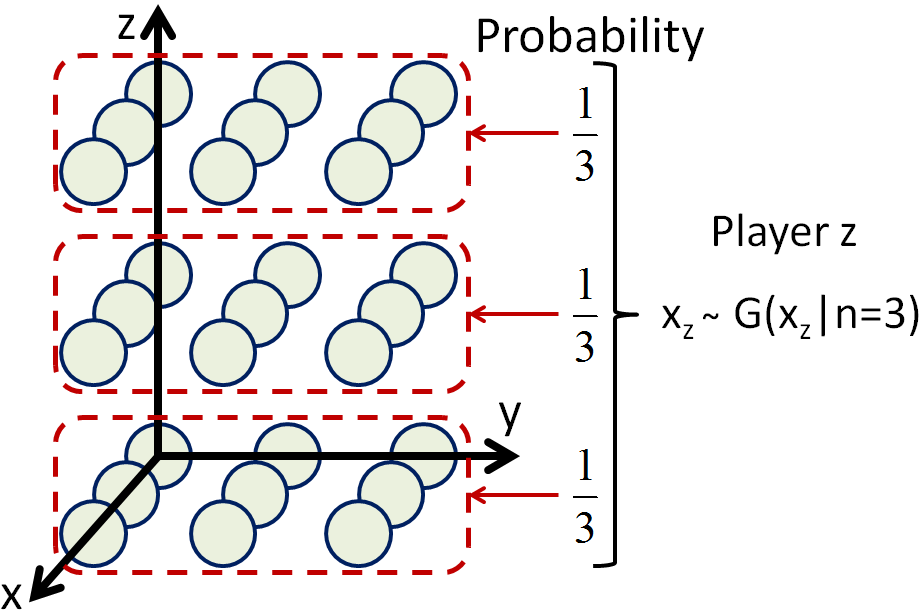} &
\includegraphics[scale=0.25]{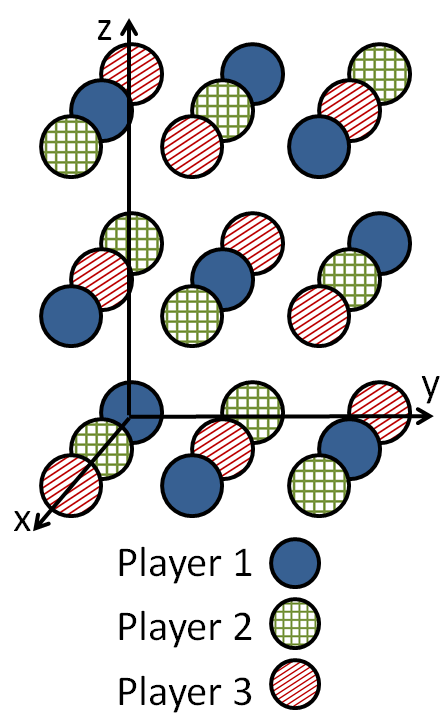}\\
$(a)$&$(b)$&$(c)$
\end{tabular}
 \caption{The figure illustrates the cases $n=2$ and $n=3$ ($(a)$ and $(b)$ respectively) for the
      lower bound example with submodular valuation functions. In $(c)$ an optimal allocation for the case $n=3$ is shown.}
    \label{fig:exampleLBsub}
\end{figure}

  Let $F(x)$ denote the probability that bidder $i$ gets a fixed item
  $j$, given that he bids $b_i(j)=x$ for this item, and the bids in
  $\b_{-i}$ are drawn from $\B_{-i}$ (due to symmetry, this
  probability is the same for all items $w=(\ell, w_{-i})$). For every
  other player $k,$ the probability that he bids $0$ for item $j$ is
  $(n-1)/n,$ and the probability that $j$ is in his selected slice but
  he bids lower than $x$ is $G(x)/n.$ Multiplying over the $n-1$ other
  players, we obtain
  \[F(x)= \left(\frac{G(x)}{n}+\frac{n-1}{n}\right)^{n-1}=
  \frac{\left(\frac{n-1}{n}\right)^{n-1}}{1-x}.\] Notice that $v_i$ is
  an additive valuation restricted to the slice of items that player
  $i$ bids for in a particular $b_i.$ Therefore, when player $i$ bids $x$ in $b_i,$ his expected utility
   is $F(x)(1-x)$ for one of these
  items, and comprising all items it is $\E[u_i(b_i)] =n^{n-1}
  F(x)\cdot(1-x)=n^{n-1} \left(\frac{n-1}{n}\right)^{n-1}=(n-1)^{n-1}$.

  Next we show that $\B$ is a Nash equilibrium. In particular, the
  bids $b_i$ in the support of $B_i$ maximize the expected utility of
  a fixed player $i.$

  First, we fix an arbitrary $w_{-i},$ and focus on the set of items
  $C:=\{(\ell,w_{-i})\,|\, \ell \in [n]\},$ which we call a
  \emph{column} for player $i.$ Recall that $i$ is interested in
  getting only one item within $C,$ while his valuation is
  additive over items from different columns. Moreover, in a fixed
  $\b_{-i}$, every other player $k$ submits the same bid for all items
  in $C,$ because either the whole $C$ is in the current slice of $k,$
  and he bids the same value $x$, or no item from the column is in the
  slice and he bids $0.$ Consider first a deviating bid, in which $i$
  bids a positive value for more than one items in $C,$ say (at least)
  the values $x\geq x'>0$ where $x$ is his highest bid in $C.$ Then
  his expected utility for this column is strictly less than
  $F(x)(1-x),$ because his value is $F(x)\cdot 1,$ but he might have
  to pay $x+x',$ in case he gets both items. Consequently, bidding $x$
  for only one item in $C$ and $0$ for the rest of $C$ is more
  profitable.

  Second, observe that restricted to a fixed column, submitting any
  bid $x\in\left[0,1-\left(\frac{n-1}{n}\right)^{n-1}\right]$ for one
  arbitrary item results in the constant expected utility of $
  \left(\frac{n-1}{n}\right)^{n-1}$, whereas a bid higher than
  $1-\left(\frac{n-1}{n}\right)^{n-1}$ guarantees the item but pays
  more so the utility becomes strictly less than $
  \left(\frac{n-1}{n}\right)^{n-1}$ for this column. In summary,
  bidding for exactly one item from each column, an arbitrary
  (possibly different) bid
  $x\in\left[0,1-\left(\frac{n-1}{n}\right)^{n-1}\right]$ is a best
  response for $i$ yielding the above expected utility, which
  concludes the proof that $\B$ is a Nash equilibrium.

  It remains to calculate the expected social welfare of $\B,$ and the
  optimal social welfare.  We define a random variable w.r.t. the
  distribution $\B.$ Let $Z_j=1$ if one of the real players $1,\ldots
  ,n$ gets item $j,$ and $Z_j=0$ if player $0$ gets the item. Note that the social welfare is the random variable $\sum_{j\in M}Z_j,$ and the
  expected social welfare is
  \[\E_{\b\sim\B}[SW(\b)]=\sum_{j}E[Z_j]\cdot 1=n^n(1-Pr(\text{no real
    player bids for j})) =
  n^n\left(1-\left(\frac{n-1}{n}\right)^{n}\right).\] 
  Finally, we show
  that the optimum social welfare is $n^n.$ An optimal allocation can
  be constructed as follows: For each item $(w_1,w_2, ...,w_n)$
  compute $r=\left(\sum_{i=1}^n w_i \mod n\right)$. Allocate this item
  to the player $r+1.$ It is easy to see that this way the $n$ items of \emph{any} particular column $\{(\ell,w_{-k})\,|\, \ell \in [n]\}$ (in any direction $k$) are given to the $n$ different players, and that each player is
  allocated $n^{n-1}$ items (Figure \ref{fig:exampleLBsub}$(c)$ shows
 the optimum allocation for $n=3$). In other words, any two items allocated to
  the same player differ in at least two coordinates. In particular,
  they belong to different columns of \emph{this player}, and all contribute
  $1$ to the valuation of the player, which is therefore $n^{n-1}.$
  Since this valuation is maximum possible for every player, the
  obtained social welfare of $n^{n}$ is optimal.

  \iffalse
  We prove the claim next. Assume for the sake of contradiction that
  there is a player who has been allocated two items that differ in
  exactly one coordinate.  Let them be $(w_1, ... , w_{j-1}, w_j,
  w_{j+1}, ... , w_n)$ and $(w_1, ... , w_{j-1}, w'_j, w_{j+1}, ... ,
  w_n)$. Since both items have been allocated to the same player, the
  following should hold: $\left(\sum_{i=1}^n w_i \mod n\right) =
  \left(\sum_{i=1, i\neq j}^n w_i + w'_j \mod n\right)$ and so,
  $\left(w_j \mod n\right) = \left(w'_j \mod n\right)$. But since
  $1\leq w_j,w'_j \leq n$, we can derive that $w_j=w'_j$, a
  contradiction.
  
  \fi
 
  Thus, the Price of Anarchy is
  $\frac{1}{\left(1-\left(\frac{n-1}{n}\right)^{n}\right)}$, and for
  large $n$ it converges to
  $\frac{1}{\left(1-\frac{1}{e}\right)}\approx 1.58$.
\end{proof}

%%% Local Variables: 
%%% mode: latex
%%% TeX-master: "poa"
%%% End: 

\section{Subadditive valuations}
\label{sec:subadditive_LB}

Here we show a lower bound of $2$ on the mixed PoA when players 
have subadditive valuations. This lower bound matches 
the upper bound by Feldman et al.~\cite{FFGL13}.

\begin{theorem} \label{thm:lower_suba} The mixed PoA of
  simultaneous first price auctions with subadditive bidder
  valuations is at least $2$.
\end{theorem}
\begin{proof}
  Consider two players and $m$ items with the following valuations:
  player $1$ is a unit-demand player with valuation $v<1$ 
	(to be determined later) if she gets
  at least one item; player $2$ has valuation $1$ for getting at least one but less
  than $m$ items, and $2$ if she gets all the items.
  Inspired by \cite{HKMN11}, we use the following distribution
  functions:

$$G(x)=\frac{(m-1)x}{1-x} \qquad x\in [0,1/m]; \qquad\qquad F(y) = \frac{v-1/m}{v-y}\qquad y\in [0,1/m]. $$
Player $1$ picks one of the $m$ items uniformly at random, and bids
$x$ for this item and $0$ for all other items. Player $2$ bids $y$ for
each of the $m$ items. The bids $x$ and $y$ are drawn from distributions with CDF $F(x)$ and $G(y),$
respectively.  In the case of a tie, the item is always allocated to
player $2$.

Let $\B$ denote this mixed bidding profile.  We are going to prove
that $\B$ is a mixed Nash equilibrium for every $v>1/m.$

If player $1$ bids any  $x$ in the range $(0,1/m]$ for the one item,
she gets the item with probability $F(x)$, since a tie appears with
zero probability. Her expected utility for $x\in (0,1/m]$ is $F(x)(v-x)=v-1/m$. 
Thus if player $1$ picks $x$ randomly
according to $G(x)$,  her utility is  $v-1/m$ (note that according to $G(x)$ she bids $0$
with zero probability). Bidding something greater than $1/m$ results in
a utility less than $v-1/m$. Regarding player $1$, it remains to show
that her utility while bidding for only one item is at least her
utility while bidding for more items. Suppose player $1$ bids $x_i$
for item $i$, $1\leq i \leq m$. W.l.o.g., assume $x_i \geq x_{i+1}$,
for $1\leq i \leq m-1$. Player $1$ doesn't get any item if and only if
$y \geq x_1.$ So, with probability $F(x_1)$, she gets at least one
item and she pays at least $x_1$. Therefore, her expected utility is
at most $F(x_1)(v-x_1)=v-1/m$, (but it would be strictly less if she 
is charged nonzero payments for other items). This means
that bidding only $x_1$ for one item and zero for the rest of them
dominates the strategy we have assumed.

If player $2$ bids a common bid $y$ for all items, where
$y \in[0,1/m]$, she gets $m$ items with probability $G(y)$ and $m-1$
items with probability $1-G(y)$. Her expected utility is
$G(y)(2-my)+(1-G(y))(1-(m-1)y)=G(y)(1-y)+1-(m-1)y=1$. %Bidding something greater than $1/m$ results in utility less than $2-m\cdot 1/m=1$. Finally 
We show that player $2$ cannot get a utility
higher than $1$ by using any deviating bids. Suppose now that player
$2$ bids $y_i$ for item $i$, for $1\leq i \leq m$.  Player $1$ bids
for item $i$ (according to $G(x)$) with probability $1/m.$ We also use that since $G$ is a CDF, for $x>1/m$ holds that $G(x)=1<\frac{(m-1)x}{1-x}.$  So, the
expected utility of player $2$ is

\begin{eqnarray*}
  &&\frac{1}{m}\sum_{i=1}^m \left( G(y_i)\left(2-\sum_{j=1}^m y_j
    \right) +(1-G(y_i))\left(1-\sum_{\substack{j=1\\j\neq i}}^m y_j\right) \right) \\
  &=& \frac{1}{m}\sum_{i=1}^m \left( G(y_i)(1- y_i )+1-\sum_{\substack{j=1\\j\neq i}}^m y_j \right)\\
  &\leq& \frac{1}{m}\sum_{i=1}^m \left(\frac{(m-1)y_i}{1-y_i}(1- y_i )+1-\sum_{\substack{j=1\\j\neq i}}^m y_j \right)\\
  &=& \frac{1}{m}\sum_{i=1}^m \left( my_i +1-\sum_{j=1}^m y_j \right)\\
  &=& \frac{1}{m}\left( m\sum_{i=1}^m y_i +m-m\sum_{j=1}^m y_j \right) = 1.
\end{eqnarray*}
Overall, we proved that $\B$ is a mixed Nash equilibrium.  

It is easy
to see that the optimal allocation gives all items  to player 2, and has social welfare $2.$ In
the Nash equilibrium $\B,$ player $2$ bids $0$ with probability
$1-\frac{1}{mv}$, so, with at least this probability, player $1$ gets
one item.

\begin{eqnarray*}
  SW(\B)\leq \left(1-\frac{1}{mv} \right)(v+1)+\frac{1}{mv} 2 =1+v+\frac{1}{mv}-\frac{1}{m}
\end{eqnarray*}
If we set $v=1/\sqrt{m}$, then $SW(\B) \leq
1+\frac{2}{\sqrt{m}}-\frac{1}{m}$. So, PoA$ \geq
\frac{2}{1+\frac{2}{\sqrt{m}}-\frac{1}{m}}$ which, for large $m$,
converges to 2.
\end{proof}

%%% Local Variables: 
%%% mode: latex
%%% TeX-master: "poa"
%%% End: 

\section{Additive valuations}
\label{sec:additive}

For additive valuations, we show that 
mixed Nash equilibria are {\em efficient}, whenever they exist. This implies an interesting 
separation between mixed equilibria with full
information and Bayesian
equilibria, that are known not to be efficient~\cite{Kri02}. For the sake of completeness, we present a 
lower bound of $1.06$ for the Bayesian PoA of  single-item auctions, in Appendix \ref{sec:single_LB_Bayes}.

\subsection{The PoA for single item auctions is $1$}
\label{sec:single_mixed}

We consider a first-price single-item auction, where the valuations
of the players for the item are given by $(v_1,v_2,\ldots,v_n).$ We
show that the PoA in mixed strategies is 1. For any mixed Nash
equilibrium of strategies $\B=(B_1,B_2,\ldots,B_n),$ let $B_i$ denote
the probability measure of the distribution of bid $b_i;$ in
particular, $B_i(I)=\P[b_i\in I]$ for any real interval $I.$ The
corresponding cumulative distribution function (CDF) of $b_i$ is
denoted by $G_i$ (i.e., $G_i(x)=B_i((-\infty,x])$). Recall that 
for a given $\B,$ for every bidder $i\in
  [n],$  $F_i(b_i)$ denotes the CDF of $max_{j\neq i} b_{j}$.
We also use $\varphi_i(b_i)$ to denote the probability that player $i$ gets the item
with bid $b_i.$ Note that $\varphi(b_i)\le F_i(b_i),$ due to a possible tie
at $b_i.$ %\footnote{Our result holds for arbitrary (randomized) tie-breaking rules.% , as long as, for any fixed $\b=(b_1,\ldots,b_n),$
  % the probabilities for the players to get the item are fixed.}
Therefore, if he bids $b_i,$ then his expected utility is
$$\E[u_i(b_i)]=\E_{\b_{-i}\sim
  \B_{-i}}[u_i(b_i)]=\varphi_i(b_i)(v_i-b_i)\le F_i(b_i)(v_i-b_i).$$
Let $\E[u_i]=\E_{\b\sim \B}[u_i]$ denote his overall expected utility defined by the (Lebesgue) integral
$\E[u_i]=\int_{(-\infty,\infty)} \varphi_i(x)(v_i-x)dB_i.$
Furthermore, assuming $\P[b_i\in I]>0$ for some interval $I,$ let
$\E[u_i| b_i\in I]=\E_{\b\sim \B}[u_i| b_i\in I]$ be the expected
utility of $i,$ on condition that his bid is in $I.$ By
definition 

\[\E[u_i| b_i\in I]=\frac{\int_{I}
  \varphi_i(x)(v_i-x)dB_i}{\P[b_i\in I]}\le\frac{\int_{I}
  F_i(x)(v_i-x)dB_i}{\P[b_i\in I]}.\]
  
The next lemma follows from the definition of (mixed) Nash equilibria. It states that in equilibrium the excepted utility of any player $i$, conditioned on the event $b_i\in I,$ must be equal to his overall expected utility, given that he bids with positive probability in the interval $I$.

\begin{lemma}\label{lem:utility} In any mixed Nash equilibrium $\B,$
  for every player $i$ holds that if $\P[b_i\in I]>0$ , then $\E[u_i|
  b_i\in I]=\E[u_i].$
\end{lemma}

\begin{proof} Assume that $\E[u_i| b_i\in I]>\E[u_i].$ Then, the
  player would be better off by submitting bids only in the interval
  $I$ (according to the distribution $B_i'(I')= B_i(I')/B_i(I)$ for
  all $I'\subset I$). If $\E[u_i| b_i\in I]<\E[u_i],$ the proof is
  analogous: in this case the player would be better off by bidding
  outside the interval. Both cases would contradict $\B$ being a Nash
  equilibrium.
  \end{proof}

In the next lemma we show that for any two players with positive
utility, the infimum of their bids must be equal, and they both bid
higher than this value with probability $1.$ Intuitively, if a player
has non-zero utility, then his lowest possible bid cannot be lower than
any player's lowest bid.

\iffalse

\begin{lemma}\label{lem:lowestbid} Assume that in a mixed Nash
  equilibrium $\B$ there are bidders $i$ and $j,$ with positive
  utilities $\E[u_i]>0,$ and $\E[u_j]>0.$ Let $q_i=\inf_x\{G_i(x)>
  0\}$ and $q_j=\inf_x\{G_i(x)> 0\}.$ Then $q_i=q_j=q,$ and
  $G_i(q)=G_j(q)=0,$ consequently $F_i(q)=F_j(q)=0.$
\end{lemma}

\begin{proof} Assume w.l.o.g. that $q_i>q_j.$ Note that, by the
  definition of $q_j,$ player $j$ bids with positive probability in
  the interval $I=[q_j,q_i).$

  On the other hand, $F_j(x)=0$ over interval $I,$ since (at least)
  player $i$ bids higher than $x$ with probability $1.$ This implies
  $\E[u_i|b_i\in I]=0.$ Using Lemma~\ref{lem:utility} we obtain
  $\E[u_i]=\E[u_i| b_i\in I]=0,$ contradicting our assumptions.  This
  proves $q_i=q_j=q.$
  
  \fi

\begin{lemma}\label{lem:lowestbid} Assume that in a mixed Nash
  equilibrium $\B$ there are bidders $i$ and $j,$ with positive
  utilities $\E[u_i]>0,$ and $\E[u_j]>0.$ Let $q_i=\inf_x\{G_i(x)>
  0\}$ and $q_j=\inf_x\{G_j(x)> 0\}.$ Then $q_i=q_j=q,$ and
  $G_i(q)=G_j(q)=0,$ consequently $F_i(q)=F_j(q)=0.$
\end{lemma}

\begin{proof} Assume w.l.o.g. that $q_i>q_j.$ Note that, by the
  definition of $q_j,$ player $j$ bids with positive probability in
  the interval $I=[q_j,q_i).$

  On the other hand, $F_j(x)=0$ over interval $I,$ since (at least)
  player $i$ bids higher than $x$ with probability $1.$ This implies
  $\E[u_j|b_j\in I]=0.$ Using Lemma~\ref{lem:utility} we obtain
  $\E[u_j]=\E[u_j| b_j\in I]=0,$ contradicting our assumptions.  This
  proves $q_i=q_j=q.$

  Next we show $G_i(q)=G_j(q)=0.$ Observe first, that because of
  $\E[u_j]>0,$ $v_j>q$ and $v_i>q$ must hold, since $q$ is the smallest possible
  bid of $j$ and of $i.$  Assume now that $G_i(q)>0$ and
  $G_j(q)=0.$ Then, $\P[b_i=q]>0,$ but $\E[u_i| b_i=q]=0,$ since $j$
  bids higher. This contradicts again Lemma~\ref{lem:utility} for the
  interval $[q,q].$

  Second, assume that $G_i(q)>0$ and $G_j(q)>0.$ In case $b_i=b_j=q,$
  bidder $i$ or bidder $j$ receives the item with probability smaller
  than 1. W.l.o.g., we assume it is player $i$. In this case bidder
  $i$ is better off by bidding $q+\epsilon$ for a small enough
  $\epsilon$ instead of bidding $q,$ since in case of bids
  $b_i=q+\epsilon,$ and $b_j=q,$ he gets the item for
  sure. % \btnote{I
  % think it is not clear that i gets the item for sure. We need to
  % add the claim that all player bids at most q}
  This contradicts $\B$ being a Nash equilibrium, and altogether we
  conclude $G_i(q)=G_j(q)=0.$

  Finally, this immediately implies $F_i(q)=F_j(q)=0,$ since for both
  $i$ and $j,$ (at least) the other one bids higher than $q$ with
  probability 1.
\end{proof}

Finally, we prove that mixed equilibria are always efficient. We use the above lemma to show that all
players who have non-zero utility, must have maximum valuation. 

\begin{theorem}\label{thm:addmixed} In a single-item auction the PoA
  of mixed Nash equilibria is 1.
\end{theorem}

\begin{proof} Let $v_i$ be the maximum valuation in the single item
  auction with full-information. Assume for the sake of contradiction
  that a mixed Nash equilibrium $\B$ has $\E_{\b\in \B}[SW(\b)]<
  SW(OPT)=v_i.$ Then, there is a nonempty set of bidders $J\subset
  [n]\setminus \{i\},$ who all get the item with \emph{positive}
  probability in $\B,$ moreover $v_k<v_i$ holds for all $k\in J.$ Let
  $j\in J$ denote the player with maximum valuation $v_j<v_i$ among
  players in $J.$
 
  We show that $\E[u_i]>0,$ and $\E[u_j]>0.$ Let us first consider the
  distribution $F_i(x)$ of the maximum bid in $\b_{-i}.$ If
  $F_i(v_i-\delta)=0$ for all $\delta>0,$ then the highest bid in
  $\b_{-i},$ and thus the payment of player $j$ is at least $v_i>v_j$
  whenever $j$ wins the item. Thus for his utility $\E[u_j]<0,$
  contradicting that $\B$ is a Nash equilibrium. Therefore, there
  exists a small $\delta,$ such that $F_i(v_i-\delta)>0.$ This implies
  $\E[u_i]>0,$ because by bidding $v_i-\delta/2$ only, player $i$
  would have higher than $0$ utility.

  Now assume for the sake of contradiction that $\E[u_j]=0$
  ($\E[u_j]<0$ is impossible in an equilibrium). Note that
  $F_j(v_j)>0,$ otherwise $j$ would get the item with positive
  probability, but always for a price higher than $v_j.$ On the other
  hand, if there were a small $\delta'$ such that
  $F_j(v_j-\delta')>0,$ then $j$ could improve his $0$ utility by
  bidding $v_j-\delta'/2$ only. The latter implies, that $j$ can get
  the item (with positive probability) only with bids $v_j$ or higher,
  so he never bids higher so as to avoid negative expected
  utility. We obtained that $\varphi_j(v_j)>0,$ where $\varphi_j(v_j)$ denotes
  the probability of $j$ winning the item with bid $v_j.$
  
  Moreover, $F_j(v_j-\delta')=0$ for all $\delta'>0$ implies
  that the minimum bid of at least one player $k$ is at least $v_j$
  ($\inf_x\{G_k(x)> 0\}\geq v_j$).  Therefore the
  \emph{winning} bids of player $i$ are also at least $v_j$ (both when
  $i=k,$ and when $i\neq k$). 
  But then $i$ could improve his utility by overbidding the (with positive probabilty) winning bid $v_j$ of $j,$ that is,
  by bidding exactly $v_j+\epsilon$ (instead of $\leq v_j+\epsilon$)
  with probability $G_i(v_j+\epsilon)$ for a small enough $\epsilon.$
  With this bid, the \emph{additional} utility of $i$ would  
  get arbitrarily close to (at least) $\varphi_j(v_j)(v_i-v_j)>0.$

  Thus, we established the existence of players $i$ and $j,$ with
  different valuations $v_j<v_i,$ and both with strictly positive
  expected utility in $\B.$ According to Lemma~\ref{lem:lowestbid},
  for the infimum of these two players' bids $q_i=q_j=q,$ and
  $F_i(q)=F_j(q)=0$ hold. Furthermore, $q<v_j<v_i,$ otherwise the
  utility of $j$ could not be positive. By the definition of $q_i=q,$
  for any $\epsilon>0$ it holds that $\P[q\leq b_i<q+\epsilon]>0.$
  Therefore, by Lemma~\ref{lem:utility}, and by the definition of
  conditional expectation, for the interval $I=[q, q+\epsilon)$ we
  have
  \begin{eqnarray*}
    \E[u_i]&=&\E[u_i|b_i\in I]\le\frac{\int_{I} F_i(x)(v_i-x)dB_i}{\P[b_i\in I]}\\
    &<& \frac{\int_{I} F_i(q+\epsilon)(v_i-q)dB_i}{\P[b_i\in I]}=F_i(q+\epsilon)(v_i-q).
  \end{eqnarray*}
  Rearranging terms, this yields $F_i(q+\epsilon)>\E[u_i]/(v_i-q)>0$
  for arbitrary $\epsilon>0.$ Since $F_i(x)$ as a cumulative distribution function is
  right-continuous in every point, this positive lower bound must hold for $F_i(q)$ as well, contradicting
  $F_i(q)=0.$
\end{proof}

\subsection{Upper bound for additive valuations}
\label{sec:additive_mixed}

We extend the above proof for additive valuations.

\begin{theorem}\label{thm:addmixedm} For simultaneous first-price
  auctions with additive valuations the PoA of mixed Nash equilibria
  is 1.
\end{theorem}
\begin{proof}
  Let $\B$ be a mixed Nash equilibrium in the $m$ item case. We argue
  first that for any fixed bidder $i,$ it is without loss of
  generality to assume that in $B_i$ his bids for each item are drawn
  from independent distributions. If this were not the case, we could
  determine the distribution $B_i^j$ of $b_i(j)$ for any item to have
  the same CDF $G_i^j$ as the distribution of bids for this item in
  $B_i.$ Then we would replace $B_i$ by the product distribution for
  the items $B_i'=\times B_i^j.$ Since both the expected valuation and
  the expected payment for item $j$ would remain the same in this new
  strategy, and the valuation and utility of the player are the sum of
  valuations and utilities over the items, none of these amounts would
  be affected. Furthermore the same additivity holds for any other
  player $k,$ whose 'price function' $F_k^j(\cdot)$ for item $j$ would
  also not be influenced. Thus, with $B_i$ replaced by the strategy
  $B_i',$ the mixed profile $\B'=(B_i',B_{-i})$ would remain a mixed
  Nash with the same expected social welfare as $\B.$

  The remaining argument is similar. Now the distribution of bids
  $(B_i^j)_{i\in [n]}$ for any particular item $j$ corresponds to a mixed Nash
  equilibrium of the single item auction for this item. Otherwise a
  player could improve his utility for $j,$ and consequently the sum
  of his utilities for all items. In turn, by Theorem~\ref{thm:addmixed} this implies that the
  social welfare for each item $j$ is optimal, a player (or players)
  with maximum valuation receive the item, which concludes the proof.
\end{proof}

%%% Local Variables:
%%% mode: latex
%%% TeX-master: "poa"
%%% End:

\section{Bid-Dependent Auctions}
\label{sec:general}

Here we generalize some of our results to simultaneous
bid-dependent auctions.  
Intuitively, a single item auction is \emph{bid-dependent} if the
winner is always the highest bidder, and a bidder's payment 
depends only on whether she gets the item or not, and on her \emph{own} bid. 
For instance, the first-price auction and the all-pay auction  
are bid-dependent but the second-price auction is not.

 For a given simultaneous bid-dependent auction, we will denote by $q_j^w(x)$ and
$q_j^l(x)$ a bidder's payment $p_{ij}(\b)$ for item $j$  when her bid for $j$ is
$x,$ depending on whether she is the winner or a loser, respectively. Note that we assume $q_j^w(x)$ (resp. $q_j^l(x)$)
to be the same for all bidders. Without this assumption the PoA is unbounded, as we show in Appendix
\ref{sec:anonymity}. To guarantee the existence of reasonable Nash
Equilibria, we also make the following natural assumptions about
$q^w_j(x)$ and $q^l_j(x)$:\footnote{Similar assumptions are also made in \cite{KR12}, \cite{LP00}
and \cite{Bre08}.}

\noindent -- $q^w_j(x)$ and $q^l_j(x)$ are
 non-decreasing, continuous functions of $x$ 
and normalized, 
such that $q^l_j(0)=q^w_j(0)=0;$ 

\noindent -- $q^w_j(x)\ge q^l_j(x)$ for all $x\ge 0;$ 

\noindent -- $q^w_k(x)>0$ for some $x$ (to avoid the case of all 
payments being zero, for that no Nash equilibria exist).

\subsection{Fractionally Subadditive valuations}

\subsubsection{Upper Bounds}
\label{sec:one_item_upperGen}

% Here we discuss a new approach to obtain bounds for the PoA in
% simultaneous first price, and discriminatory auctions. While the upper
% bounds that we derive with the help of this idea, can also be
% obtained based on the very general \emph{smoothness} framework
% \cite{Rou12,ST13,KMST13}, the approach we introduce here does not
% adhere to this framework.\footnote{Roughly, because the pure deviating
%   bid $a$ that we identify, depends on the other players' bids
%   $\b_{-i}$ in the NE.} The strength of our approach consists in its
% potential to lead to better (in some cases tight) \emph{lower} bounds,
% as we demonstrate in the subsequent sections.

In this section we discuss the general upper bound for simultaneous
bid-dependent auctions. 

We define $\theta$ as  
$\theta=\max_{j\in[m]}\sup_{\{x:q^w_j(x)\neq 0\}} \{q_j^l(x)/q^w_j(x)\}$. 

 Observe that $\theta \in [0,1],$  due to the assumption $q_j^l(x)\le q^w_j(x).$ 
We will prove that (for $\theta\neq 1$) the coarse-correlated and the
  Bayesian PoA of simultaneous bid-dependent auctions with fractionally subadditive bidders is at
  most $\frac{(\theta-1)^2}{\theta^2-\theta+1-e^{\theta-1}}.$ When we set $\theta = 0$ or $\theta \rightarrow 1,$ we get back the upper bounds of $e/(e-1)$ for 
	first-price auctions, and $2$ for all-pay auctions, respectively.

%For bidders with fractionally subadditive valuations, we are able to
%show the PoA is between $1.58$ and $2$. % with another assumption that
% for any item $j$, rank $r\geq 2$ and $x\geq 0$,
% $q_j(x,r)=q_j(x,2)$. To put it in words, the bidders should have the
% same payment functions if she is not a winner. For simplicity, we will
% use $q_j^w(x)$ and $q_j^l(x)$ to denote the bidder's payment for item
% $j$ when her bid is $x$ and her is winner and loser respectively.

%For any random bidding profile $\B$,  We skip the dependency of $\B$ in $F_{ij}$, since $\B$ will be always clear from the context. Similar to Section \ref{sec:one_item_upper}, we are able

We start by  proving a lemma for a single item, analogous to Lemma \ref{lem:1.58upper1}.

\begin{lemma}
\label{lem:submodularGen} Consider a single item bid-dependent auction with payment functions $q^w(x)$ and $q^l(x).$
Let $\B$ be an arbitrary randomized bidding profile, and  $F_{i}$ denote the CDF of the random variable $\max_{k\neq i}b_{k}$, for this $\B.$ 
Then for every bidder $i$, and
non-negative value $v$, there exists a pure bidding strategy $a=a(v,\B_{-i})$ such that, 
  $$F_{i}(a)\left(v-q^w(a)+q^l(a)\right)-q^l(a)+\sum_{k\in[n]}p_{k}(\B)\ge
  \frac{\theta^2-\theta+1-e^{\theta-1}}{(\theta-1)^2}\cdot v,$$
where $p_{k}(\B)=\E_{\b\sim \B} [p_k(\b)]$ is the expected payment from player $k.$
  
\end{lemma}

\begin{proof}
  Let $a = \arg\max_x\left\{F_{i}(x)\left(v-q^w(x)+q^l(x)\right)-q^l(x)\right\}$ and $A =
  F_{i}(a)\left(v-q^w(a)+q^l(a)\right)-q^l(a)$. 
  %
  %We show that $v\geq A\geq 0.$
	%Observe that 
	%$A =  F_{i}(a)v-F_{i}(a)q^w(a)-(1-F_{i}(a))q^l(a) \leq v$, since $F_i$ is a CDF; moreover 
	%$A \geq F_{i}(0)\left(v-q^w(0)+q^l(0)\right)-q^l(0)=F_i(0)v\geq 0.$
  %
  In the following we use that  $F_i$ is the CDF of $\max_{k\neq i}b_{k}$, and since $q^w(\cdot)$ is continuous,  $\E[q^w(\max_{k\neq i}b_{k})]=\int_0^\infty(1-F_i(x)) dq^w(x)$ holds. %Moreover, for non-decreasing $q^w(x)$ also $ \int_0^\infty(1-F_i(x)) dq^w(x)\geq \int_0^{v-A}(1-F_i(x)) dq^w(x).$

  \begin{eqnarray*}
    A + \sum_k p_k(\B) &\ge& A+\E_{\b}[q^w(\max_{k\neq i}b_{k})]\notag\\
    &=& A+\int_0^{\infty}(1-F_{i}(x))dq^w(x)\notag\\
    &\geq& A+\int_0^{\infty}\left(1-\frac{A+q^l(x)}{v-q^w(x)+q^l(x)}\right)dq^w(x)\\
    &\geq& A+\int_0^{\infty}\left(\frac{v-A-q^w(x)}{v+(\theta - 1)q^w(x)}\right)dq^w(x)\\
    &\geq& A+\int_0^{v-A}\left(\frac{v-A-y}{v+(\theta - 1)y}\right)dy \notag
  \end{eqnarray*}
The second inequality follows from the definition of $A$ and $a$ and
  the third one is due to the fact that $q^l_j(x) \leq \theta \cdot
  q^w_j(x)$ for any $x$. For the last one, $q^w(0) = 0$ and we further need to 
	show that for $x_0 = \infty$, $q^w(x_0) \geq v-A$: by definition 
	$A \geq F_{i}(x_0)\left(v-q^w(x_0)+q^l(x_0)\right)-q^l(x_0) = v-q^w(x_0)$ since $F(x_0)=1$, 
	meaning that $q^w(x_0) \geq v-A$. For completeness we also show that $v-A \geq 0$, by showing 
	that $v\geq A\geq 0$: observe that 
	$A =  F_{i}(a)v-F_{i}(a)q^w(a)-(1-F_{i}(a))q^l(a) \leq v$, since $F_i$ is a CDF; moreover 
	$A \geq F_{i}(0)\left(v-q^w(0)+q^l(0)\right)-q^l(0)=F_i(0)v\geq 0.$

  %?? I changed the proof above because I think this argument is not correct. %For the last inequality, notice that for $\bar{x}$ s.t. $F_{ij}(\bar{x})=1$, $F_{ij}(\bar{x})(v-q_w(\bar{x})+q_l(\bar{x}))-q_l(\bar{x}) \leq A$,  so $q_w(\bar{x}) \geq v-A$ and therefore the bounds of integral are valid.--> I cannot see why it wasn't correct, however I agree with your changes, it is more elegant now and we don't need the explanation that you commented out here.

% [?? The technical details of the next calculations can be found in Appendix?? --> I just used maple to compute the integral. The result is stated below. Then again with maple, I found the derivative and the point that it becomes 0, which is the value of A that minimizes the expression. I am not sure what to put in the appendix. I added the maple file in the dropbox under the name ComputationsUBgeneralsubmod (I forgot some term in the expression below for thela < 1; I corrected it) 
%??Do you think we should put the integral the minimization and the limit in an appendix??]

  In case $\theta < 1$, $A + \sum_k p_{kj}(\B) \ge A+
  \frac{(A+\theta(v-A))(\ln (A+\theta(v-A))-\ln (v)) - (\theta -
    1)(v-A)}{(\theta-1)^2}$, which is minimized for $A =
  \frac{v(\theta \cdot e^{1-\theta}-1)}{(\theta
    -1)e^{1-\theta}}$. The lemma follows by replacing $A$ with this value.

  In case $\theta = 1$, $A + \sum_k p_{kj}(\B) \geq
  A+\frac{(v-A)^2}{2v} \geq \frac 12 v$. The
  limit of $\frac{\theta^2-\theta+1-e^{\theta-1}}{(\theta-1)^2}$ when
  $\theta \rightarrow 1$ is $\frac 12.$  
\end{proof}

In the following, let $f^S_i(\cdot)$ be a maximizing additive
  function of set $S$ for player $i$ with fractionally subadditive valuation function
  $v_i$. By the definition of fractionally subadditive valuations, we have that $v_i(T) \geq f^S_i(T)$,
  for every $T \subseteq S$ and $f^S_i(S) = v_i(S)$.

\begin{lemma}\label{lem:XOSpropertyGen}
  For any set $S$ of items, and any strategy profile $\b$, where $b_{ij} = 0$ for $j\notin S$,
  $$u_i(\b) \geq \sum_{j \in S} \Big(\P[j\in X_i(\b)]\left(f^S_i(j)-q_j^w(b_{ij})+q_j^l(b_{ij})\right) - q_j^l(b_{ij})\Big).$$
\end{lemma}

\begin{proof}
  
  \begin{eqnarray*}
    u_i(\b) &\geq& \sum_{T \subseteq S}  \P[X_i(\b)=T]\left(f^S_i(T)-
      \sum_{j\in T}q_j^w(b_{ij})-\sum_{j\in S\smallsetminus
        T}q_j^l(b_{ij})\right)  \\
    &=& \sum_{T \subseteq S}\sum_{j\in
      T}\P[X_i(\b)=T]\left(f^S_i(j)-q_j^w(b_{ij})\right)-\sum_{T
      \subseteq S}\sum_{j\in S\smallsetminus
      T}\P[X_i(\b)=T]q_j^l(b_{ij})\\
    &=& \sum_{j\in S}\sum_{\substack{T \subseteq S:j \in T}} \P[X_i(\b)=T]\left(f^S_i(j)-q_j^w(b_{ij})\right)
		-\sum_{j\in S}\sum_{\substack{T \subseteq S:j \notin T}}\P[X_i(\b)=T]q_j^l(b_{ij})\\
    &=& \sum_{j\in S} \P[j\in X_i(\b)]\left(f^S_i(j)-q_j^w(b_{ij})\right) - \sum_{j\in S} \P[j\notin X_i(\b)]q_j^l(b_{ij})\\
		&=& \sum_{j \in S} \Big(\P[j\in X_i(\b)](f^S_i(j)-q_j^w(b_{ij})+q_j^l(b_{ij})) - q_j^l(b_{ij})\Big).
  \end{eqnarray*}
\end{proof}

%??check? --> OK
\begin{lemma}\label{lem:tie_XOSGen}
Let $\B$ be a coarse correlated equilibrium of a simultaneous bid-dependent auction.  For any set of items $S$ and any pure strategy $b'_i$ of player $i$, where $b'_{ij}=0$ for $j\notin S$,
  \[u_i(\B)=\E_{\b\sim\B}[u_i(\b)]\geq \sum_{j\in
    S}\Big(F_{ij}(b'_{ij})\left(f^S_i(j)-q^w_j(b'_{ij})+q^l_j(b'_{ij})\right) - q^l_j(b'_{ij}) \Big).\]
\end{lemma}

The proof of the lemma is analogous to that of Lemma \ref{lem:tie}: $\E_{\b\sim \B}[u_i(\b)]\geq \E_{\b\sim \B}[u_i((b_i',\b_{-i}))]$ holds since $\B$ is an equilibrium; then Lemma \ref{lem:XOSpropertyGen} is applied to $(b_i',\b_{-i}),$ and expectation is taken over $\b\sim\B.$

\begin{theorem}
  \label{thm:generalpay_mixedGen} For  bidders with
  fractionally subadditive valuations, the coarse correlated PoA of any
  bid-dependent auction is at most
  $\frac{(\theta-1)^2}{\theta^2-\theta+1-e^{\theta-1}}.$ 
\end{theorem}
\begin{proof} Let $\B$ be a coarse correlated equilibrium, and 
  let $\lambda(\theta) =
  \frac{\theta^2-\theta+1-e^{\theta-1}}{(\theta-1)^2}.$ For every player $i,$ consider the 
  maximizing additive valuation, $f^{O_i}_i$ for his optimal set $O_i$. By Lemma
  \ref{lem:submodularGen},  for every fixed player $i$ and item $j$ there exists a bid $a_{ij}$ such that
  $$F_{ij}(a_{ij})\left(f^{O_i}_i(j)-q^w_j(a_{ij})+q^l_j(a_{ij})\right)-q^l_j(a_{ij})\ge
  \lambda(\theta)f^{O_i}_i(j) - \sum_kp_{kj}(\B)$$ 
  For player $i$, we consider the deviation that her bid is $a_{ij}$ for every item in $O_i$ (and $0$ for all other items), and apply Lemma \ref{lem:tie_XOSGen}. Combined with the above inequality (for all items in $O_i$), we obtain  
  $$u_i(\B) \geq \lambda(\theta)  \sum_{j \in O_i}f^{O_i}_i(j) - \sum_{j \in O_i}\sum_kp_{kj}(\B)\\
    = \lambda(\theta)v_i(O_i^{\v}) - \sum_{j \in O_i}\sum_kp_{kj}(\B).$$
  By summing over all players, we get
  $$\sum_iu_i(\B) \geq \lambda(\theta)\sum_iv_i(O_i^{\v}) - \sum_{j\in [m]}\sum_kp_{kj}(\B)=\lambda(\theta)SW(\O)- \sum_kp_k(\B)$$
    The theorem follows from $SW(\B)=\sum_iu_i(\B)+ \sum_ip_i(\B).$
\end{proof}

Similarly to Lemmas \ref{lem:tie} and \ref{lem:tie_XOSGen}, we can prove
the following.

\begin{lemma}\label{lem:tie_BayesGen} Assume that $\B$ be is a  Bayesian Nash
  equilibrium, and let $S$ be an arbitrary set of items. For  player $i$ with valuation $v_i,$ let $b'_i$
  be a pure strategy such that $b'_{ij} = 0$ for $j\notin S$. Then,
  \[\E_{\substack{\v_{-i}\\\b\sim \B(\v)}}[u_i^{v_i}(\b)]\ge
  \sum_{j\in S}\left(F^{v_i}_{ij}(b'_{ij})(f^S_{v_i}(j)-q^w_j(b'_{ij})+q^l_j(b'_{ij}))-q^l_j(b'_{ij}) \right).\]
\end{lemma}

\begin{theorem}
  \label{thm:generalpay_bayesGen} The Bayesian PoA of any bid-dependent
  auction, when the bidders have fractionally subadditive and
  independently distributed valuations, is at most
  $\frac{(\theta-1)^2}{\theta^2-\theta+1-e^{\theta-1}}$.
\end{theorem}

%?? w_i is (just a bit) problematic because of q^w --> I changed it to v'

\begin{proof}
 Let $\lambda(\theta) =
  \frac{\theta^2-\theta+1-e^{\theta-1}}{(\theta-1)^2}$.
Suppose that $\B$ is a Bayesian Nash Equilibrium and the valuation of
  each player $i$ is drawn according to $v_{i}\sim D_{i},$ where the $D_i$ are independently
  distributed. We use the notation $\C=(C_1,C_2,\ldots,C_n)$ to denote
  the bidding distribution in $\B$ which involves the randomness of the valuations $\v,$ and of the
  bidding strategy $\B(v),$  that is $b_i(v_i)\sim
  C_i$. Then the utility of player $i$ with valuation $v_i$ can be
  expressed by $u_i(\B_i(v_i), \C_{-i})=\E_{b_i\sim\B_i(v_i),\b_{-i}\sim \C_{-i}}[u_i(\b)]$. It should be noted that
  $\C_{-i}$ does not depend on a particular $v_{-i}$ (just on the distribution $\D$). Also notice that the following equality holds:
  $\E_{\v_{-i}}[u^{v_i}_i(\B_i(v_i),\B_{-i}(\v_{-i}))] =
  u^{v_i}_i(\B_i(v_i),\C_{-i})$ \footnote{$\E_{\v_{-i}}[u_i^{v_i}(\B(\v))]=\E_{b_i \sim B_i(v_i)} \E_{\substack{\v_{-i} \\ \b_{-i} \sim \B_{-i}(\v_{-i})}}[u_i^{v_i}(b_i,b_{-i})]=\E_{b_i \sim B_i(v_i)} \E_{\b_{-i} \sim \C_{-i}}[u_i^{v_i}(b_i,b_{-i})]=u^{v_i}_i(\B_i(v_i),\C_{-i})$}.

  %[?? 1. I don't understand 'sampling' and the whole following sentence. $a_{ij}$ should depend on $w_{-i}$, does it? 2. I also changed two $\B$-s to $\C$ below, please check.] --> a_{ij} depends on the distribution C_{-i} and therefore on the D_{-i} not on any sampling valuation. I chang it to be more clear.
  
  For any player $i$ and any fractionally subadditive valuation $v_i\in V_i$, consider the
  following deviation: consider some $\v'_{-i}\sim D_{-i}$ and then for every
  $j\in O(v_i,\v'_{-i})$ bid $a_j(v_i,\C_{-i})$ as defined in Lemma
  \ref{lem:submodularGen}.  By applying Lemma \ref{lem:tie_BayesGen} for
  $S=O_i(v_i,\v'_{-i}),$ taking expectation over $v_i$ and $\v'_{-i}$ 
  and summing over all players, we have that
  
  \begin{align*}
   &\sum_i \E_{\v}[u^{v_i}_i(\B(\v))]\\&=\sum_i \E_{\v}[u^{v_i}_i(\B_i(v_i),\C_{-i})]\\
	&\geq\sum_{i} \E_{\substack{v_i,\v'_{-i}}} \left[ \sum_{j\in
        O_i(v_i,\v'_{-i})}\left(F_{ij}^{v_i}(a_{ij})\left(
        f_{v_i}^{O_i(v_i,\v'_{-i})}(j) - q^w_j(a_{ij})+q^l_j(a_{ij})\right) -q^l_j(a_{ij})\right)\right]\\
 &= \sum_{i} \E_{\v'} \left[ \sum_{j\in
       O_i(\v')}\left(F_{ij}^{v'_i}(a_{ij})\left(
        f_{v'_i}^{O_i(\v')}(j) - q^w_j(a_{ij})+q^l_j(a_{ij})\right) -q^l_j(a_{ij})\right)\right]\\
 &\geq \sum_{i} \E_{\v'} \left[ \sum_{j\in
       O_i(\v')}\left(\lambda(\theta)\cdot f_{v'_i}^{O_i(\v')}(j)-\sum_k p_{kj}(\B_i(v_i),\C_{-i})\right)\right] \\
    &= \lambda(\theta) \cdot\sum_i\E_\v[v_i(O_i^{\v})]-\sum_i\sum_jp_{kj}(\C)
  \end{align*}
  The last inequality follows by Lemma
 \ref{lem:submodularGen}. \\
 So, $\E_{\v}[SW(\B(\v))]= \sum_i\E_{\v}[u_i(\B)]+ \sum_i\sum_jp_{kj}(\C)\ge
  \lambda(\theta)\cdot \E_{\v}[SW(\O^{\v})]$.
\end{proof}

%%% Local Variables: 
%%% mode: latex
%%% TeX-master: "poa"
%%% End: 

\subsubsection{Lower Bound}
\label{sec:xos-LB}

Here we present a lower bound of ${\frac{e}{e-1}}$ for the PoA 
of simultaneous bid-dependent auctions with OXS valuations and
for mixed equilibria. This implies a lower bound for submodular and
fractionally subadditive valuations, as well as for more general classes of
equilibria.

% \begin{theorem}
% \label{1.58LB} The price of anarchy of simultaneous first price
%   auctions with full information and OXS valuations is at least
%   ${\frac{e}{e-1}}\approx 1.58.$
% \end{theorem}
\begin{theorem}
\label{1.58LBGen} The PoA of simultaneous bid-dependent
  auctions with full information and OXS valuations is at least
  ${\frac{e}{e-1}}\approx 1.58.$
\end{theorem}

\begin{proof}
The proof is very similar to the one for simultaneous first price auctions (Section \ref{sec:submodular_LB}). 
Therefore, here we only point out the differences. 
 The same construction applies here; the only difference appears in the Nash strategy
  profile and in a scaling of the valuations.

  We  choose an appropriate value $V$, such that
  $V\left(1-\left(\frac{n-1}{n}\right)^{n-1}\right)$ is in the range
  of $q^w_j(\cdot)$ for all $j$ (notice that due to our assumptions
  on $q^w_j$, there exists such a $V$).
  We consider the same set of players and items as in Section~\ref{sec:submodular_LB};
  %??Theorem~\ref{thm:submodular_LB}
  the valuation functions of the players are the same as in Section~\ref{sec:submodular_LB}, 
	except that each valuation is
    multiplied by $V.$ Also, the same tie breaking rule applies.
  
  As for the mixed Nash equilibrium $\B$, the dummy player still bids $0$ 
	for every item and every real player still picks an $n-1$ dimensional slice 
	in the same random way. However the bid $x_j$ that she bids for every 
  item $j$ of that slice is drawn according to a distribution with the following 
	{\em item-specific} CDF 
	(we will show below that $G_j$ is a valid CDF):
  $$G_j(x)=n\left(\frac{V\left(\frac{n-1}{n}\right)^{n-1} +
      q^l_j(x)}{V-q^w_j(x)+q^l_j(x)}\right)^{\frac{1}{n-1}}-n+1, \qquad
  x\in\left[0,T_j\right]$$ 
  where $T_j$ is the bid such that
  $q^w_j(T_j)=V\left(1-\left(\frac{n-1}{n}\right)^{n-1}\right)$. 
  %In order to guarantee that $\B$ is a mixed Nash equilibrium, we further need those bids to be correlated. 
  Notice that we can no longer require that 
	the bids of a player on different items are equal, since the CDFs $G_j$ are different.  
	Instead, we require that for 
	every real player the bids $x_j$ for different 
	items in her slice are correlated in the 
	following way: she chooses $\rho$
  uniformly at random from the interval $[0, 1],$ and then sets $x_j=
  G_j^{-1}(\rho)$, for every $j$ in her slice. Note that for
  any two items $j_1$, $j_2$ of the slice, it holds that
  $G_{j_1}(x_1) = G_{j_2}(x_2)=\rho$ and $x_{j_1}$ is not necessarily equal to 
	$x_{j_2}$. However, for each item $j$ in the slice, the way
    that $x_j$ is chosen is equivalent to sampling it according to the CDF
    $G_j(x_j)$ (but in a correlated way to the other bids). The fact that  
		each player's bids are such that the CDF values become equal, will be sufficient 
		for proving that $\B$ is a mixed Nash equilibrium. 
  
  The probability $F_j(x)$ that a player gets item $j$ if she bids $x$ for it is:
  $$F_j(x) = \left( \frac{G_j(x)}{n} +
  \frac{n-1}{n}\right)^{n-1}=\frac{V\left(\frac{n-1}{n}\right)^{n-1} +
  q^l_j(x)}{V-q^w_j(x)+q^l_j(x)}, \quad x\in\left[0,T_j\right]$$

  Recall that the valuation of player $i$ is additive, restricted to the slice
  of items  that she bids for in a particular $b_i.$
  Therefore the expected utility of $i$ when he bids $x$ in $b_i$ for item $j$
  is $F_j(x)(V-q^w_j(x)) - (1-F_j(x))q^l_j(x) =
  F_j(x)(V-q^w_j(x)+q^l_j(x)) - q^l_j(x) =
  V\left(\frac{n-1}{n}\right)^{n-1}$.  
  By comprising all items, $\E[u_i(b_i)] = Vn^{n-1} \left(\frac{n-1}{n}\right)^{n-1}$.

	%Proving that $\B$ is a Nash equilibrium needs more of attention, so we prove it 
	%separately in the next claim.

\begin{claim}
$\B$ is a Nash equilibrium.
\end{claim}

\begin{proof}
First, we fix an arbitrary $w_{-i}\in [n]^{n-1},$ and focus on the set of items
  $C:=\{(\ell,w_{-i})\,|\, \ell \in [n]\},$ which we call a
  \emph{column} for player $i.$ Recall that $i$ is interested in
  getting only one item within $C,$ on the other hand his valuation is
  additive over items from different columns. Moreover, in a fixed
  $\b_{-i}$, every other player $k$ submits bids $x_j$ resulting in 
	equal values of $G_j(x_j)$ for all items
  in $C,$ because either the whole $C$ is in the current slice of $k,$
  and he bids correlated bids on them, or no item from the column is in the
  slice and he bids $0.$ 
  
  %??I changed, please check
  Consider first a deviating bid, in which $i$
  bids a positive value for more than one items in $C,$ say (at least)
  the values $x_1, x_2>0$ for items $j_1$, $j_2$, respectively, 
	and w.l.o.g. assume that $G_{j_1}(x_1)$ is maximum over items in $C.$  We prove that if she 
	loses item $j_1$ she should lose item $j_2$ as well: if she loses $j_1,$ then there must be a bidder $k$ with bid $x_1'>x_1$
  for item $j_1.$ Since $G_{j_1}(x)$ is increasing, this implies $G_{j_1}(x_1')>G_{j_1}(x_1).$ However, since  the bids of player $k$ are correlated (and $j_2$  is in his slice as well), for his bid $x_2'$ on $j_2$ it holds that $G_{j_2}(x_2')=G_{j_1}(x_1')>G_{j_1}(x_1)\geq G_{j_2}(x_2).$ Therefore, $x_2'>x_2,$ so player $i$ cannot win item $j_2$ either, so 
  %Let $y_1$ and $y_2$ be the 
	%highest bids among the rest of the bidders for items $j_1$, $j_2$, 
	%respectively. If $y_1 \geq x_1$, then 
	%(based on the correlated bids of the other players) $F_{j_2}(y_2) = F_{j_1}(y_1) 
	%\geq F_{j_1}(x_1) \geq F_{j_2}(x_2)$. The tie braking 
	%rule and the continuity of $F_j$ in $(0,T_j]$ indicate that if  
bidding for item $j_2$ cannot contribute to the valuation, whereas the 
	bidder might pay for more items than $j_1$. 
	Consequently, bidding 
  for only one item in $C$ and $0$ for the rest of $C$ is more
  profitable.
	
  Second, observe that restricted to a fixed column, submitting any
  bid $x\in\left[0,T_j\right]$ for one
  arbitrary item $j$ results in the constant expected utility of $
  V\left(\frac{n-1}{n}\right)^{n-1}$, whereas by bidding higher than
  $T_j$  the utility would be at most $V-q^w_j(T_j)=
  V\left(\frac{n-1}{n}\right)^{n-1}$ for this column. In summary,
  bidding for exactly one item $j$ from each column, an arbitrary bid
  $x\in\left[0,T_j\right]$ is a best
  response for $i$ yielding the above expected utility, which
  concludes the proof that $\B$ is a Nash equilibrium.
\end{proof}

  The rest of the argument is exactly the same as in the proof for first price auctions, 
  with both $SW(\B)$ and $SW(\O)$ scaled by $V,$ that cancels out in the price of anarchy.

  It remains to prove that the $G_j(\cdot)$ are valid
  cumulative distribution functions for every $j$. To this end it is sufficient
  to show that $G_j(T_j)=1$ and that $G_j(x)$ is non-decreasing in
  $[0,T_j]$. For simplicity we skip index $j$.

  $$G(T)=n\left(\frac{V\left(\frac{n-1}{n}\right)^{n-1} +
      q^l(T)}{V-q^w(T)+q^l(T)}\right)^{\frac{1}{n-1}}-n+1 =$$
  $$=n\left(\frac{V\left(\frac{n-1}{n}\right)^{n-1} +
      q^l(T)}{V-V\left(1-\left(\frac{n-1}{n}\right)^{n-1}\right)+q^l(T)}\right)^{\frac{1}{n-1}}-n+1=1$$
  Now let $x_1,x_2 \in [0, T_j],$ and $x_1>x_2.$ In order to prove $G(x_1)\geq G(x_2)$, it is sufficient to prove that
  $\frac{V\left(\frac{n-1}{n}\right)^{n-1} +
    q^l(x_1)}{V-q^w(x_1)+q^l(x_1)}\ge
  \frac{V\left(\frac{n-1}{n}\right)^{n-1} +
    q^l(x_2)}{V-q^w(x_2)+q^l(x_2)}$.

  \begin{eqnarray*}
    &&\frac{V\left(\frac{n-1}{n}\right)^{n-1} +
      q^l(x_1)}{V-q^w(x_1)+q^l(x_1)}-
    \frac{V\left(\frac{n-1}{n}\right)^{n-1} +
      q^l(x_2)}{V-q^w(x_2)+q^l(x_2)} \\ 
    % &=&\frac{V\left(\frac{n-1}{n}\right)^{n-1}(q^w(x_1)-q^l(x_1)-q^w(x_2)+q^l(x_2))
    %   + q^l(x_1)(V-q^w(x_2))-
    %   q^l(x_2)(V-q^w(x_1))}{(V-q^w(x_1)+q^l(x_1))(V-q^w(x_2)+q^l(x_2))}\\
    &=&\frac{q^l(x_1)(V-V\left(\frac{n-1}{n}\right)^{n-1}-q^w(x_2))-
      q^l(x_2)(V-V\left(\frac{n-1}{n}\right)^{n-1}-q^w(x_1))}{(V-q^w(x_1)+q^l(x_1))(V-q^w(x_2)+q^l(x_2))}\\
    &\ge&\frac{(q^l(x_1)-
      q^l(x_2))(V-V\left(\frac{n-1}{n}\right)^{n-1}-q^w(x_1))}{(V-q^w(x_1)+q^l(x_1))(V-q^w(x_2)+q^l(x_2))}\ge 0
  \end{eqnarray*}
The last two inequalities follow from the monotonicity of $q^l$ and
  $q^w,$ and from the fact that
  $V\left(1-\left(\frac{n-1}{n}\right)^{n-1}\right)\geq q^w(x_1)$ holds by
  the definition of $T_j$.
\end{proof}

%%% Local Variables: 
%%% mode: latex
%%% TeX-master: "poa"
%%% End: 

\subsection{Subadditive valuations}
\label{sec:subGen}

We prove tight bounds for the PoA in simultaneous bid-dependent auctions
with subadditive bidder valuations. We show that the coarse-correlated and the
  Bayesian PoA is exactly $2.$ Our results hold even 
for a class of auctions more general than bid-dependent auctions: 
 we allow the payment rule to depend on the \emph{rank} of the bid, where the
$r^{th}$ highest bid (with an arbitrary tie-breaking rule) has rank $r$. 
We use $q_j(x,r)$ to denote the payment that the bidder should pay for
item $j$ when her bid is $x$ and gets rank $r$. In particular, given a
mixed bidding strategy $\B$, bidder $i$'s expected payment $p_i(\B)$ is 
equal to $\sum_{j\in[m]}\E_{\b}[q_j(b_{ij},r_i(\b^j)]$ where
$r_i(\b^j)$ denotes the rank of $b_{ij}$ among
$\{b_{1j},\ldots,b_{nj}\}$. That is, $p_{ij}(\B)=\E_{\b}[q_j(b_{ij},
r_i(\b^j)]$. Note that $q^w_j(x)$ from the previous subsection is
$q_j(x,1)$ here, and $q^l_j(x)$ can be different for different
ranks. 
%(Later in Theorem \ref{thm:generalpayLBSub}, we keep the notation $q^w_j(x)$ and $q^l_j(x)$ for the payment of winning 
%and losing bids, respectively, because there are only two players in the lower bound construction. 
Analogous assumptions to the ones made on $q^w_j(x)$ and $q^l_j(x)$ can be
made on $q_j(x,r)$ as well. For the following upper bound we only assume that the $q_j(.,r)$ are normalized and increasing, and
$q_j(x,1)\ge q_j(x,r).$ %?? check  

\subsubsection{Upper Bounds}
\label{sec:subadd-UBGen}

%?? I changed a_i to A_i

\begin{lemma}
\label{lem:subadditivegeneralpay}
For any simultaneous bid-dependent auction, subadditive valuation profile $\v$ and
randomized bidding profile $\B$, there exists a randomized bid vector
$A_i(\v,\B_{-i})$ for each player $i$, such that for the total  expected utility and expected payments of the bidders
  $$\sum_iu_i(A_i(\v,\B_{-i}), \B_{-i})\ge \frac 12 \sum_iv_i(O_i^{\v})-\sum_i\sum_{j}p_{ij}(\B)$$
  holds, where $O_i^{\v}$ is the optimal set of player $i.$
\end{lemma}

\begin{proof}
  Under the profile $\v$, $O_i^{\v}$ is the set of items allocated to
  player $i$ in the optimum. We denote by $h_j(\b) = \arg \max_i
  b_{ij}$ the bidder with the highest bid for item $j$, regarding the
  pure bidding $\b$. Let $t_{ij}$ be the maximum of bids for item $j$
  among players other than $i,$ and $t_i$ be the vector such that its
  $j^{th}$ coordinate equals  $t_{ij}$ if $j\in O_i^{\v},$ and $0$
  otherwise. Note that $t_i\sim T_i$ is an induced random variable
  of $\B_{-i}$. We define the randomized bid
  $A_i(\v,\B_{-i})$ to follow the same distribution $T_i$ 
  (inspired by \cite{FFGL13}).
  
  We use the notation $v_i(b_i,t_i)$ and
  $W_i(b_i,t_i)$ to denote the player $i$'s valuation and winning
  set when she bids $b_i$ and the prices are $t_i$, i.e., $v_i(S)$
  and $W_i(S)$ where $S=\{j|b_{ij}\ge t_{ij}\}$. 
  \begin{eqnarray*}
    &&u_i(A_i(\v,\B_{-i}), \B_{-i})\\
                    &=&\E_{a_i\sim A_i}\E_{\b_{-i}\sim \B_{-i}}[u_i(a_i,\b_{-i} )]\\
		&\geq &\E_{a_i\sim T_i}\E_{t_i\sim
      T_i}[v_i(a_i,t_i)]-\sum_{j\in O_i^{\v}}\E_{a_i\sim
      T_i}[q_j(a_{ij},1)] \qquad\mbox{(since $q_j(x,1)\ge q_j(x,r)$)}\\
    &=&\E_{t_i\sim T_i}\E_{a_i\sim
      T_i}[v_i(t_i,a_i)]-\sum_{j\in O_i^{\v}}\E_{t_i\sim
      T_i}[q_j(t_{ij},1)] \qquad\mbox{(swap $t_i$ and $a_i$)}\\
    &=& \frac{1}{2} \E_{t_i\sim T_i}\E_{a_i\sim
      T_i}[v_i(t_i,a_i)+v_i(a_i,t_i)]-\sum_{j\in O_i^{\v}}\E_{t_i\sim T_i}[q_j(t_{ij},1)]\\
    &\ge& \frac{1}{2}v_i(O_i^{\v})-\sum_{j\in O_i^{\v}} \E_{\b\sim \B}[q_j(b_{h_j(\b)}(j),1)]\\
		&\ge& \frac{1}{2}v_i(O_i^{\v})-\sum_{j\in O_i^{\v}} \E_{\b\sim
      \B}[q_j(b_{h_j(\b)}(j),1)]\\
      &-&\sum_{j\in O_i^{\v}} \E_{\b\sim
      \B}\left[ \sum_k q_j(b_k(j),r_k(\b^j))-q_j(b_{h_j(\b)}(j),r_{h_j(\b)}(\b^j))\right]\\
		&=& \frac{1}{2}v_i(O_i^{\v})-\sum_{j\in O_i^{\v}}\sum_{k}
    \E_{\b\sim
      \B}[q_j(b_k,r(k,\b))]=\frac{1}{2}v_i(O_i^{\v})-\sum_{j\in
      O_i^{\v}}\sum_{k} p_{kj}(\B)
  \end{eqnarray*}
  In the second inequality $v_i(t_i,a_i)+v_i(a_i,t_i)\geq v_i(O_i^{\v})$ is due to the subadditivity of $v_i;$ for $q_j(t_{ij},1)\leq q_j(b_{h_j(\b)}(j),1)$ we use that $q_j(.,1)$ is non-decreasing, and
  the fact that $t_{ij}\leq b_{h_j(\b)}(j)$, since in
  $b_{h_j(\b)}(j)$, for computing the maximum we also consider player
  $i$. For the last inequality, notice that $\sum_k
  q_j(b_k(j),r_k(\b^j))-q_j(b_{h_j(\b)}(j),r_{h_j(\b)}(\b^j)) \geq 0$,  since from the sum of all payments for item $j$ we subtracted the payment of the winner. %as payments of non winning bids in $O_i^{\v}$. 
  The lemma follows by
  summing over all players.
\end{proof}

\begin{theorem}
  \label{thm:subadditivegeneral_mixed} For bidders with subadditive valuations, the coarse correlated PoA of
  any bid-dependent auction is at most $2$.
\end{theorem}

\begin{proof}
  Suppose $\B$ is a coarse correlated equilibrium (notice that $\v$ is
  fixed). By Lemma \ref{lem:subadditivegeneralpay} and the definition
  of coarse correlated equilibrium, we have
  $$\sum_iu_i(\B)\ge \sum_iu_i(A_i(\v,\B_{-i}), \B_{-i})\ge
  \frac 12 \sum_iv_i(O_i^{\v})-\sum_i\sum_{j}p_{ij}(\B).$$
  By rearranging the terms,
  $SW(B)=\sum_iu_i(\B)+\sum_i\sum_{j}p_{ij}(\B)\ge 1/2\cdot SW(\O)$.
\end{proof}

\begin{theorem}
  \label{thm:subadditivegeneral_bayes} For bidders with subadditive and
  independent valuations, the Bayesian PoA of any
  bid-dependent auction is at most $2$.
\end{theorem}

\begin{proof}
Suppose $\B$ is a Bayesian Nash Equilibrium and the valuation of
  each player $i$ is $v_{i}\sim D_{i}$, where the $D_i$ are independently
  distributed. We use the notation $\C=(C_1,C_2,\ldots,C_n)$ to denote
  the bidding distribution in $\B$ which includes the randomness of the valuations $\v,$ and of
  the bidding strategy $\b,$ that is $b_i(v_i)\sim
  C_i$ (like in the proof of Theorem \ref{thm:generalpay_bayesGen}).  Then
  the utility of player $i$ with valuation $v_i$ can be expressed by
  $u_i(B_i(v_i), \C_{-i})$. For any player $i$ and any subadditive
  valuation $v_i\in V_i$, consider the following deviation: sampling
  $\v'_{-i}\sim D_{-i}$ and bidding $A_i((v_i,\v'_{-i}),\C_{-i})$ as
  defined in Lemma \ref{lem:subadditivegeneralpay}. By the definition
  of Nash equilibrium, we have
  $\E_{\v_{-i}}[u^{v_i}_i(B_i(v_i),\B_{-i}(\v_{-i}))]\ge \E_{\v'_{-i}}[u^{v_i}_i(A_i((v_i,\v'_{-i}),\C_{-i}), \C_{-i})]$.
  By taking expectation over $v_i$ and summing over all players,

  \begin{align*}
    \sum_i \E_{\v}[u_i(\B(\v))]&\geq\sum_i\E_{v_i,\v'_{-i}}[u^{v_i}_i(A_i((v_i,\v'_{-i}),\C_{-i}),\C_{-i})]\\
    &=\E_{\v}\left[\sum_iu_i^{v_i}(A_i(\v,\C_{-i}),
      \C_{-i})\right]\mbox{(by relabeling $\v'_{-i}$ by $\v_{-i}$)}\\
    &\ge \frac 12\cdot\sum_i\E_\v[v_i(O_i^{\v})]-\sum_i\sum_j\E_\v[p_{kj}(\B(\v))]
  \end{align*}
  where the inequality follows by Lemma~\ref{lem:subadditivegeneralpay}.
  We obtained  $\E_{\v}[SW(\B(\v))]= \sum_i\E_{\v}[u_i(\B)]+ \sum_i\sum_j\E_{\v}[p_{kj}(\B(\v))]\ge
  1/2\cdot \E_{\v}[SW(\O^{\v})].$
\end{proof}

%The proofs of these two theorems are identical to the proofs of
%Theorem \ref{thm:allpay_mixed} and Theorem \ref{thm:allpay_bayes} by
%replacing Lemma \ref{lem:all-pay} by Lemma \ref{lem:generalpay}.

%%% Local Variables: 
%%% mode: latex
%%% TeX-master: "poa"
%%% End: 

\subsubsection{Lower Bound}
\label{sec:subadd-LBGen}

\begin{theorem}
  \label{thm:generalpayLBSub} For bidders with subbaditive valuations, the mixed PoA of simultaneous bid-dependent auctions
   is at least $2$.
\end{theorem}
\begin{proof}
  We consider $2$ players and $m$ items. Let $v$ and $V$ (with $v<V$) be positive reals
  to be defined later. Player
  $1$ has value $v$ for every non-empty subset of items; 
  player $2$ values with $V$ any non-empty strict subset of the items and
  with $2V$ the whole set of items. 
	Consider now the mixed strategy profile $\B$, where player $1$ picks item
  $l$ uniformly at random and bids $x_l$ for it and $0$ for the rest of the items,
  whereas, player $2$ bids $y_j$ for every item $j$. For $1\leq j \leq
  m$, $x_j$ and $y_j$ are drawn from distributions with the following CDFs
  $G_j(x)$ and $F_j(y)$, respectively:
  $$G_j(x)=\frac{(m-1)q^w_j(x)+q^l_j(x)}{V-q^w_j(x)+q^l_j(x)},
  \;\;x\in [0,T_j]; \quad\quad F_j(y) =
  \frac{v-V/m+q^l_{j}(y)}{v-q^w_j(y)+q^l_{j}(y)},\;\;y\in [0,T_j], $$
  where $T_j$ is the bid such that $q^w_j(T_j)=V/m$. We choose $V$
  such that $V/m$ is in the range of $q^w_j(\cdot)$ for all $j$
  (notice that due to the assumptions on $q^w_j$, there always exists
  such a value $V$). Furthermore, in $\B,$  the $y_j$'s are correlated in
  the following way: player $2$ chooses $\rho$ uniformly at random
  from the interval $[0, 1]$ and if $\rho \in \left[0,\frac{v-V/m}{v}\right)$ then 
	$y_j = 0$, otherwise $y_j= F_j^{-1}(\rho)$, for every
  $1\leq j \leq m$.\footnote{For each item $j$,  the way
    player $2$ chooses $y_j$ is equivalent to picking it according to the CDF
    $F_j(y)$.}  Note that for every two items $j_1$, $j_2$, it holds 
that $F_{j_1}(y_{j_1}) = F_{j_2}(y_{j_2})$.
%the probability that player $1$ gets item $j_1$ when she bids $x$ for this item equals the probability that she gets item $j_2$ when she bids $x$ for item $j_2$, for every items $j_1$ and $j_2$. 
  In case of a tie, player $2$  gets the item. Due to the
  continuity of $q^w_j$ and $q^l_j$, $G_j(x)$ and $F_j(x)$ are continuous
  and therefore none of the CDF have a mass point in any $x\neq 0$.
  
 We show below that $\B$ is a Nash equilibrium, and each of the $F_j$ and $G_j$ are valid cumulative distributions. 
The PoA can then be derived as follows.
 Player $2$ bids
  $0$ with probability $1- \frac V{mv}$ so, 
  $\E[SW(\B)] \leq \left(1- \frac V{mv} \right) (V+v) + \frac V{mv} 2V = V+v +\frac {V^2}{mv} - \frac Vm. $
  For $v=V/\sqrt{m}$, PoA $\geq
  \frac{2V}{V+\frac{2V}{\sqrt{m}}-\frac{V}{m}} =
  \frac{2}{1+\frac{2}{\sqrt{m}}-\frac{1}{m}}$ which, for large $m$,
  converges to $2$.

\begin{claim}
	\label{lem:SubLBNash}
  $\B$ is a Nash equilibrium. 
	\end{claim}

\begin{proof} If player
  $1$ bids any $x$ in the range of $(0,T_j]$ for a single item $j$ and
  zero for the rest, her utility is
  $F_j(x)(v-q^w_j(x))+(1-F_j(x))(-q^l_{j}(x))=F_j(x)(v-q^w_j(x)+q^l_{j}(x))-q^l_{j}(x)=v-V/m$. Since
  $G(0)=0$,  her utility is also $v-V/m$ if
  she bids according to $G(\cdot)$. Suppose player $1$ bids
  $x=(x_1,\ldots ,x_m)$, $(x_j \in [0,T_j])$ 
	with at least two positive bids. W.l.o.g., assume $F_1(x_1) =\max_i
  F_i(x_{i})$. If $y_1 \geq x_1$, player $1$ doesn't get any item,
  since for every $j$, $F_j(y_j) = F_1(y_1) \geq F_1(x_1) \geq
  F_j(x_j)$ and so $y_j \geq x_j$ (recall that in any tie player $2$
  gets the item). If $y_1 < x_1$, player $1$ gets at least the first
  item and has valuation $v$, but she cannot pay less than
  $q^w_1(x_1)$. So, this strategy is dominated by the strategy of
  bidding $x_1$ for the first item and zero for the rest. Bidding
  $x_j>T_j$ for any item guarantees the item but results in a
  utility less than $v-q^w_j(x_j)\leq v-q^w_j(T_j)=v-V/m$, so it is dominated by the strategy of
  bidding exactly $T_j$ for this item.
	
	If player $2$ bids $(y_1,\ldots, y_m)$ for every item $j$ so that $y_j\in [0,T_j]$, then (since player 1 bids positive for any particular item $j$ with probablility $1/m$)  her expected utility is  

  \begin{eqnarray*}
    &&\frac{1}{m}\sum_{j=1}^m \left(
      G_j(y_j)(2V-\sum_{\substack{k=1}}^m
      q^w_{k}(y_k))+(1-G_j(y_j))(V-\sum_{\substack{k=1\\k\neq j}}^m
      q^w_k(y_k)-q^l_{j}(y_j)) \right) \\
    &=&\frac{1}{m}\sum_{j=1}^m \left(V+ G_j(y_j)
      (V-q^w_{j}(y_j)+q^l_{j}(y_j)) -\sum_{\substack{k=1\\k\neq j}}^m
      q^w_k(y_k)-q^l_{j}(y_j)\right) \\
    &=&\frac{1}{m}\sum_{j=1}^m \left(V+ (m-1)q_j^w(y_j)+q_j^l(y_j) -\sum_{\substack{k=1\\k\neq j}}^m
      q^w_k(y_k)-q^l_{j}(y_j)\right) \\      
    &=&\frac{1}{m}\left( mV+m\sum_{j=1}^m q^w_j(y_j) -m\sum_{k=1}^m q^w_k(y_k)
    \right) = V.
  \end{eqnarray*}
  Bidding greater than $T_j$ for any item is dominated by the strategy
  of bidding exactly $T_j$ for this item.
  Overall, $\B$ is Nash equilibrium.\end{proof}

 \begin{claim}
	\label{lem:SubLB_CDF}
  $G_j(\cdot)$ and $F_j(\cdot)$ are valid cumulative distributions. 
	\end{claim}

 \begin{proof}   
	It is sufficient to show that for every $j,$ $G_j(T_j)=F_j(T_j)=1$ and 
  $G_j(x)$ and $F_j(x)$ are non-decreasing in $[0,T_j].$ In the following  we skip index $j$.

  $$G(T)=\frac{(m-1)q^w(T)+q^l(T)}{V-q^w(T)+q^l(T)} = \frac{(m-1)\frac Vm+q^l(T)}{V-\frac Vm+q^l(T)}=1,$$ $$F(T) = \frac{v-V/m+q^l(T)}{v-q^w(T)+q^l(T)}=\frac{v-V/m+q^l(T)}{v-V/m+q^l(T)}=1,$$
  Now let $x_1 >x_2$, $x_1,x_2 \in [0, T_j].$
  
  \begin{eqnarray*}
    &&G(x_1)-G(x_2)\\
    &=&\frac{(m-1)q^w(x_1)+q^l(x_1)}{V-q^w(x_1)+q^l(x_1)}-\frac{(m-1)q^w(x_2)+q^l(x_2)}{V-q^w(x_2)+q^l(x_2)}\\
    &=&
    \frac{V(m-1)(q^w(x_1)-q^w(x_2))+m(q^w(x_1)q^l(x_2)-q^l(x_1)q^w(x_2))+
      V(q^l(x_1)-q^l(x_2))}{(V-q^w(x_1)+q^l(x_1))(V-q^w(x_2)+q^l(x_2))}\\
    % &=& \frac{V(m-1)(q^w(x_1)-q^w(x_2))+m((q^w(x_1)-q^w(x_2))
    % q^l(x_2)-(q^l(x_1)-q^l(x_2))q^w(x_2))+V(q^l(x_1)-q^l(x_2))}{(V-q^w(x_1)+q^l(x_1))(V-q^w(x_2)+q^l(x_2))}\\
    &=& \frac{(V(m-1)+m\cdot
      q^l(x_2))(q^w(x_1)-q^w(x_2))+m(\frac{V}{m}-q^w(x_2))(q^l(x_1)-q^l(x_2))}{(V-q^w(x_1)+q^l(x_1))(V-q^w(x_2)+q^l(x_2))}
    \geq 0\\
  \end{eqnarray*}

  \begin{eqnarray*}
    &&F_j(x_1)-F_j(x_2)\\
    &=& \frac{v-V/m+q^l(x_1)}{v-q^w(x_1)+q^l(x_1)}-\frac{v-V/m+q^l(x_2)}{v-q^w(x_2)+q^l(x_2)}\\
    &=& \frac{(v-\frac Vm )(q^w(x_1)-q^w(x_2)) +\frac Vm
      (q^l(x_1)-q^l(x_2))+q^w(x_1)q^l(x_2)-q^l(x_1)q^w(x_2)}{(v-q^w(x_1)+q^l(x_1))(v-q^w(x_2)+q^l(x_2))}\\
    % &=& \frac{(v-\frac 1m )(q^w(x_1)-q^w(x_2)) +\frac 1m
    % (q^l(x_1)-q^l(x_2))+(q^w(x_1)-q^w(x_2))q^l(x_2)-(q^l(x_1)-q^l(x_2))q^w(x_2)}{(v-q^w(x_1)+q^l(x_1))(v-q^w(x_2)+q^l(x_2))}\\
    &=& \frac{(v-\frac Vm +q^l(x_2))(q^w(x_1)-q^w(x_2)) +(\frac Vm -
      q^w(x_2))(q^l(x_1)-q^l(x_2))}{(v-q^w(x_1)+q^l(x_1))(v-q^w(x_2)+q^l(x_2))}
    \geq 0\\
  \end{eqnarray*}
  For both inequalities we use the monotonicity of
  $q,$ moreover that $q^w_j(x)\leq V/m$ for $x\in[0,T_j],$ and
  $v=V/\sqrt{m}$ hold.
  \end{proof}\end{proof}

%%% Local Variables: 
%%% mode: latex
%%% TeX-master: "poa"
%%% End: 

\section{Discriminatory auctions}
\label{sec:discriminatory}

% First of all, we prove a useful lemma for XOS functions of identical
% elements.
% \begin{lemma}
%   \label{lem:XOS_dis}
%   For any XOS valuation function $v$ of identical items,
%   $v(n_1)/n_1\ge v(n_2)/n_2$ for all integer $n_1<n_2 \in [m]$.
% \end{lemma}

% \begin{proof}
%   Since the class of XOS functions is equivalent to fractional
%   subadditive functions, we have $v(S)\le \sum_j\alpha_iv(T_i)$ with
%   $0\le\alpha_i\le 1$ for all $i\in[m]$ if for any $j\in S$,
%   $\sum_{i|j\in T_i}\alpha_i\ge 1$ \cite{Fei06}. It is not hard to
%   see that any set $S$ with size $n_2$ can be fractional covered by
%   $n_2$ sets where each sets has size $n_1$ and
%   $\alpha_i=1/n_1$. That is
%   \[v(n_2) \le \sum_{i\in [n_2]}\alpha_iv(T_i) = n_2\cdot
%   v(n_1)/n_1\]
% \end{proof}

Discriminatory auctions are
\emph{multi-unit} auctions, i.e. $m$ units of the same item are sold
to $n$ bidders. We denote the valuation of player $i$ for $j$ units of
the item by $v_i(j).$ The valuation $v_i$ is \emph{submodular}, if the
items have decreasing marginal values, that is, $v_i(s+1)-v_i(s)\geq
v_i(t+1)-v_i(t)$ holds if $s\leq t.$ It is called \emph{subadditive},
if $v_i(s+t)\leq v_i(s)+v_i(t).$

We assume a \emph{standard} multi-unit auction in which each player
submits a vector $b_i$ of $m$ decreasing bids $b_i(1)\geq
b_i(2)\geq\ldots\geq b_i(m)\geq 0.$ The bidding profile of all players
is then $\b=(b_1,b_2,\ldots,b_n).$ In the allocation
$\xi(\b)=(\xi_1,\xi_2,\ldots,\xi_n),$ bidder $i$ gets $\xi_i$ units of
the item, if $\xi_i$ of his bids were among the $m$ highest bids of
the players. In the case of \emph{discriminatory pricing,}
every bidder $i$ pays the sum of his \emph{winning} bids, i.e. his
$\xi_i$ highest bids. 

In this section, we complement the
results by de Keijzer et
al.~\cite{KMST13} for the case of subadditive valuations, by providing a matching
lower bound of $2$ for the standard bidding format. For the case of submodular valuations, we 
provide a lower bound of $1.109$. 
We could reprove their upper bound
of $e/(e - 1)$ for submodular bids, using our non-smooth approach. 
Due to the different nature of this auction, the proof is not identical 
with the one for the first-price auction. Therefore, we present the complete
proof of this upper bound.

\subsection{Preliminaries}

 The \emph{social welfare} of the allocation $\xi(\b)$ is
$SW(\b)=SW(\xi(\b))=\sum_{i=1}^n v_i(\xi_i).$ The players have
quasi-linear utility
functions: $$u_i(\b)=v_i(\xi_i)-\sum_{j=1}^{\xi_i} b_i(j)$$

Similarly to item bidding auctions, having a \emph{mixed strategy}
$B_i,$ means that $b_i$ is drawn from the set of all possible
decreasing bid vectors according to the distribution $B_i,$ which we
denote by $b_i\sim B_i.$
Given a valuation profile $\v$ of the players, an \emph{optimal allocation}
$\o(\v) = \o^{\v} = (\o_1^{\v},\ldots \o_n^{\v})$ is one that maximizes $\sum_{i=1}^n v_i(o_i^{\v}).$

Consider a discriminatory auction with submodular valuations, with $n$
players and $m$ items. Recall that $v_i(j)$ denotes the valuation of
player $i$ for $j$ copies of the item. For any player $i$, we define
$v_{ij}=\frac{v_i(j)}{j}$. It is easy to see that for submodular
functions, $v_{ij} \geq v_{i(j+1)}$ for all $j \in [m-1]$. Let
$\beta_j(\b)$ be the $j^{th}$ \emph{lowest} bid among the \emph{winning} bids under
the strategy profile $\b.$  
Consider any randomized bidding profile
$\B = (B_1,...,B_n)$. For this $\B,$
$\beta_j(\b)$ is a random variable depending on $\b\sim\B$. We define the following
functions:
\begin{eqnarray*}
F_{ij}(x) &=& \P[\beta_j(\b_{-i}) \leq x] \qquad\qquad\qquad\qquad\qquad\qquad\qquad\qquad\qquad\;\text{for }1\leq j \leq m,\\
G_{ij}(x) &=& \P[\beta_j(\b_{-i}) \leq x < \beta_{j+1}(\b_{-i})]=F_{ij}(x) - F_{i(j+1)}(x) \qquad\quad \text{for }1\leq j \leq m-1.
\end{eqnarray*}
We define separately $G_{im}(x) = \P[\beta_m(\b_{-i}) \leq
x]=F_{im}(x).$ Notice that $F_{ij}(x)$ is the CDF of $\beta_j(\b_{-i});$ moreover for the $G_{ij}$ holds that 

\begin{eqnarray}
F_{ij}(x) &=& \sum_{k=j}^m G_{ik}(x), \notag\\
\sum_{j=1}^{m'} F_{ij}(x) &=& \sum_{j=1}^{m'} jG_{ij}(x) + \sum_{j=m'+1}^m m'G_{ij}(x). \label{sumFG}
\end{eqnarray}

We further define $F_i^{\text{av}}(x) = \frac{1}{o_i^{\v}}\sum_{\substack{j=1}}^{o_i^{\v}}F_{ij}(x),$ 
and let $\beta_i^{\text{av}}$ be a random variable with $F_i^{\text{av}}(x)$ as CDF.
$F_i^{\text{av}}(x)$ is a cumulative distribution function defined on
$\mathbb{R}^+$, since $F_i^{\text{av}}(0)=0$, $\lim_{x \to +\infty}
(F_i^{\text{av}}(x))=1$ and $F_i^{\text{av}}(x)$ is the average of
non-decreasing functions, so it is itself a non-decreasing
function. Moreover, 
\begin{eqnarray*}
\E[\beta_i^{\text{av}}] &=& \int_0^\infty
(1-F_i^{\text{av}}(x)) dx = \int_0^\infty
(1-\frac{1}{o_i^{\v}}\sum_{\substack{j=1}}^{o_i^{\v}}F_{ij}(x)) dx\\
& =& \frac{1}{o_i^{\v}}\sum_{\substack{j=1}}^{o_i^{\v}}\int_0^\infty \left(1-F_{ij}(x)\right) dx
= \frac{1}{o_i^{\v}}\sum_{j=1}^{o_i^{\v}} \beta_j(\B_{-i}),
\end{eqnarray*}
where $\beta_j(\B_{-i})=\E_{\b_{-i}\sim\B_{-i}} [\beta_j(\b_{-i})].$
Note that the above functions depend on some randomized bidding profile 
$\B_{-i}$ and on $\v$. These will be clear from context when we use 
these functions below.

\subsection{Upper bound for submodular valuations}
\label{sec:discriminatory_upper_coarse}

\begin{lemma}
\label{lem:submodularUpper1.58Discr}
For any submodular valuation profile $\v$ and any randomized bidding profile 
$\B$, there exists a pure bidding strategy $\a_i(\v,\B_{-i})$ for each player $i$, 
such that:
$$\sum_{i=1}^n u_i(\a_i(\v,\B_{-i}),\B_{-i})\geq \left(1-\frac{1}{e}\right)\sum_{i=1}^n v_i(o_i^{\v})-
\sum_{j=1}^m \beta_j(\B),$$ where $\beta_j(\B)=\E_{\b\sim \B}[\beta_j(\b)].$
\end{lemma}

\begin{proof} Recall that $v_{ij}=\frac{v_i(j)}{j}.$ Let $a_i$ be the 
value that maximizes $(v_{io_i^{\v}}-a_i)F_i^{\text{av}}(a_i)$.
Let $\a_i(\v,\B_{-i})=(\underbrace{a_i,
    \ldots ,a_i}_{o_i^{\v}},\underbrace{0,\ldots,0}_{m-o_i^{\v}})$ be the selected strategy
  profile for player $i$. Observe that by the definition of $G_{ij}(),$
	$G_{ij}(a_i)$ is the probability of $a_i$ being the $j^{th}$ 
lowest bid among winning bids under $\B_{-i}$. Therefore, if player $i$ bids 
according to $\a_i(\v,\B_{-i})$, $G_{ij}(a_i)$ is the 
probability of player $i$ getting exactly $j$ items, if $j \leq o_i^{\v}$, and 
$o_i^{\v}$ items, if $ j > o_i^{\v}$, under 
the bidding profile $(\a_i(\v,\B_{-i}),\B_{-i})$.
 Similarly to Lemma \ref{lem:tie}, we get

\begin{eqnarray*}
 u_i(\a_i(\v,\B_{-i}),\B_{-i}) 
&\geq& \sum_{\substack{j=1}}^{o_i^{\v}} G_{ij}(a_i)(v_i(j)-ja_i) + \sum_{j=o_i^{\v}+1}^m G_{ij}(a_i) (v_i(o_i^{\v})-o_i^{\v}a_i)\\
&=& \sum_{\substack{j=1}}^{o_i^{\v}} jG_{ij}(a_i)(v_{ij}-a_i) + \sum_{j=o_i^{\v}+1}^m o_i^{\v} G_{ij}(a_i) (v_{io_i^{\v}}-a_i)\\
&\geq& (v_{io_i^{\v}}-a_i)\left( \sum_{\substack{j=1}}^{o_i^{\v}} jG_{ij}(a_i) + \sum_{j=o_i^{\v}+1}^m o_i^{\v} G_{ij}(a_i)\right)\\
	&=&  (v_{io_i^{\v}}-a_i)\sum_{\substack{j=1}}^{o_i^{\v}}F_{ij}(a_i)
	= o_i^{\v} (v_{io_i^{\v}}-a_i)F_i^{\text{av}}(a_i) \\
	&\geq&  \left(1-\frac{1}{e}\right)o_i^{\v} v_{io_i^{\v}} - o_i^{\v} \E \left[\beta_i^{\text{av}}\right] 
	= \left(1-\frac{1}{e}\right)  v_i(o_i^{\v}) - o_i^{\v} \E \left[\beta_i^{\text{av}}\right]\\
	&=& \left(1-\frac{1}{e}\right)  v_i(o_i^{\v}) - \sum_{j=1}^{o_i^{\v}} \beta_j(\B_{-i})
\end{eqnarray*}
For the second inequality, $v_{ij} \geq v_{io_i}$ for submodular valuations and for the following equality, 
we used \eqref{sumFG} where $m'$ is set to $o_i^{\v}$. 
For the last inequality we apply Lemma \ref{lem:1.58upper1}, since $a_i$ maximizes the expression $(v_{io_i^{\v}}-a_i)F_i^{\text{av}}(a_i).$ 

For any pure strategy profile $\b$ 
and any valuation profile $\v$ it holds that 
$$  \sum_{j=1}^m \beta_j(\b) \geq \sum_{i=1}^n \sum_{j=1}^{o_i^{\v}} \beta_j(\b) 
	\geq \sum_{i=1}^n \sum_{j=1}^{o_i^{\v}} \beta_j(\b_{-i}). \label{pricesDis}$$
By summing up over all players and using this inequality the lemma follows.
\end{proof}

%\subsubsection{Upper bound for Bayesian equilibria with submodular valuations}
%\label{sec:discriminatory_bayes_upper}

\begin{theorem}
  The coarse correlated PoA for the discriminatory auction is at most $\frac{e}{e-1}$, 
	when the players' valuations are submodular.
\end{theorem}

\begin{proof}
Suppose $\B$ is a coarse correlated equilibrium (in this case $\v$ is fixed). 
By Lemma \ref{lem:submodularUpper1.58Discr} and the definition of coarse 
correlated equilibrium, we have that
\begin{eqnarray*}
\sum_{i=1}^n u_i(\B) &\geq& \sum_{i=1}^n u_i(\a_i(\v,\B_{-i}),\B_{-i})\\
&\geq& \left(1-\frac{1}{e}\right)\sum_{i=1}^n v_i(o_i^{\v}) - 
\sum_{j=1}^m \beta_j(\B)
\end{eqnarray*}

After rearranging the terms
$SW(\B) = \sum_i u_i(\B) + \sum_j \beta_j(\B) \geq \left(1-\frac{1}{e}\right) SW(\o).$
\end{proof}

\begin{theorem}
  The BPoA of the discriminatory auction is at most $\frac{e}{e-1}$, 
	when the players' valuations are submodular.
\end{theorem}

\begin{proof}
Suppose $\B$ is a Bayesian Nash Equilibrium and the valuation of
  each player $i$ is $v_{i}\sim D_{i}$, where the $D_i$'s are
  independent distributions. We denote by $\C=(C_1,C_2,\ldots,C_n)$
  the bidding distribution in $\B$ which includes the randomness of
  both the bidding strategy $\b$ and of the valuations $\v$, that is $b_i(v_i)\sim
  C_i$. Then the utility of agent $i$ with valuation $v_i$ can be
  expressed by $u_i(\B_i(v_i), \C_{-i})$. It should be noted that
  $\C_{-i}$ depends on $D_{-i}$ but {\em not} on the $\v_{-i}$. For any
  agent $i$ and any submodular valuation $v_i\in V_i$, consider the
  following deviation: sample $\v'_{-i}\sim \D_{-i}$ and bid
  $a_i((v_i,\v'_{-i}),\C_{-i})$ as defined in Lemma
  \ref{lem:submodularUpper1.58Discr}. By the definition of the Bayesian Nash equilibrium, we have
	%\footnote{$\E_{\v_{-i}}[u_i^{v_i}(\B(\v))]=\E_{b_i \sim B_i(v_i)} \E_{\substack{\v_{-i} \\ \b_{-i} \sim \B_{-i}(\v_{-i})}}[u_i^{v_i}(b_i,b_{-i})]=\E_{b_i \sim B_i(v_i)} \E_{\b_{-i} \sim \C_{-i}}[u_i^{v_i}(b_i,b_{-i})]=u^{v_i}_i(\B_i(v_i),\C_{-i})$}

  \[\E_{\v_{-i}}[u^{v_i}_i(\B_i(v_i),\B_{-i}(\v_{-i}))]\ge \E_{\v'_{-i}}[u^{v_i}_i(a_i((v_i,\v'_{-i}),\C_{-i}), \C_{-i})]\]
  By taking expectation over $v_i$ and summing up over all agents, 
  \begin{align*}
    &\sum_{i=1}^n \E_{\v}[u_i(\B(\v))]\\
    \geq&\sum_{i=1}^n\E_{v_i,\v'_{-i}}[u^{v_i}_i(a_i((v_i,\v'_{-i}),\C_{-i}),\C_{-i})]\\
    =&\E_{\v}\left[\sum_{i=1}^nu_i^{v_i}(a_i(\v,\C_{-i}),
      \C_{-i})\right]\mbox{(by relabeling $\v'_{-i}$ by $\v_{-i}$)}\\
    \ge& \left(1-\frac{1}{e}\right)\sum_{i=1}^n v_i(o_i^{\v})-
\sum_{j=1}^m \beta_j(\C) \\
=& \left(1-\frac{1}{e}\right)\sum_{i=1}^n v_i(o_i^{\v})-
\sum_{j=1}^m \E_{\v}[\beta_j(\B(\v))]
  \end{align*}
  So, $\E_{\v}[SW(\B(\v))]= \sum_i\E_{\v}[u_i(\B(\v))]+ \sum_j \E_{\v}[\beta_j(\B(\v))]\ge
  \left(1-\frac{1}{e}\right) \E_{\v}[SW(\o^{\v})].$
\end{proof}

\subsection{Lower bounds}
\label{sec:discriminatory_lower}

%% #

\subsubsection{Submodular valuations}
\label{sec:discriminatory_lower_submod}

\begin{theorem} The price of anarchy for submodular discriminatory
  auctions is at least $1.099.$
\end{theorem}

\begin{proof} We present an example for a discriminatory auction with
  submodular valuations and show that the PoA of mixed Nash equilibria
 is at least 1.099.

  We design a game with two players and two identical items. Player 1
  has valuation $(v,v),$ i.e., her valuation is $v$ if she gets  one or two items;
  whereas player 2 has valuation $(1,2),$ i.e., he is additive with value $1$ for each item. 
  We use the following distribution functions defined
  by Hassidim et al.~\cite{HKMN11}:

$$G(x)=\frac{x}{1-x}, \qquad x\in [0,1/2]; \qquad\qquad F(y) = \frac{v-1/2}{v-y},\qquad y\in [0,1/2]. $$

Consider the following mixed strategy profile. Player $1$ bids $(x,0)$
and player $2$ bids $(y,y)$, where $x$ and $y$ are drawn from $G(x)$
and $F(y)$, respectively.  Noting that player $2$ bids 0 with
probability $F(0)=1-1/2v$, we need a tie-breaking rule for the case of bidding $0$,
 in which player $2$ always gets the item. We claim that this mixed
strategy profile is a Nash equilibrium.

First we prove that playing $(x,0)$ for player $1$ is a best response for every
$x\in[0,1/2]$. Notice that $(x,x')$ with $x'\leq x$, is dominated by
$(x,0)$, since if player $1$ gets at least one item, she should pay at
least $x$ and getting both items doesn't add to her utility.

$$u_{1}(x,0)=F(x)\cdot(v-x)=v-1/2.$$
Clearly, bidding higher than $1/2$ guarantees the item but the payment is higher. 

Now we need to
show that $(y,y)$ is a best response for player $2$, for every $y \in [0,1/2]$.  Consider any
strategy $(y,y')$ with $y,y'\in[0,1/2]$ and $y\ge y'$.

\begin{eqnarray*}
  u_{2}(y,y')&=&\P[x\leq y'](2-y-y')+\P[x> y'](1-y)\\
  &=&G(y')(2-y-y')+(1-G(y'))(1-y)= G(y')(1-y')+1-y=1+y'-y\leq 1,
\end{eqnarray*}
and $u_2(y,y)=1$ is maximum possible. Bidding strictly higher than $1/2$ for both
items is not profitable, since then her utility is $2-2y<1.$ 

Now we calculate the expected social welfare of this Nash
equilibrium.
\begin{align*}
  \E[SW]&=\P[y\geq x]2+\P[x>y](1+v)\\
  &=2-(1-v)\P[x>y]\\
  &=2-(1-v)\int_{0}^{1/2}F(x)dG(x)
\end{align*}
This expression is maximized for $v=0.643$. For this value of $v$,
$\E[SW]= 1.818$. Since $SW(\O) = 2$, we get PoA$=1.099$.\footnote{This
  bound can be improved to 1.109 by a similar construction with 3
  items. For large $m$ the bound goes to one, so we do not believe that
  this construction is tight. That is the reason why we present here
  the simplest version of 2 items.}
\end{proof}

\subsubsection{Subadditive valuations}
\label{sec:discriminatory_lower_subadd}

We provide a tight lower bound of $2$ for subadditive valuations
 in discriminatory auctions which is similar to the lower bound of Section 
\ref{sec:subadditive_LB}, adjusted to discriminatory auctions.

\begin{theorem} \label{thm:lower_suba} For discriminatory auctions the price of anarchy in mixed Nash equilibria is at least $2$ for bidders with subadditive valuations.
\end{theorem}
\begin{proof}
  Consider two players and $m$ items with
  the following valuations: player $1$ is a unit-demand player with
  valuation $v<1$ if she gets at least one item; player $2$ has
  valuation $1$ for getting less than $m$ but at least one items, and
  $2$ if she gets all the items.  Inspired by \cite{HKMN11}, we use
  the following distribution functions:

$$G(x)=\frac{(m-1)x}{1-x}, \qquad x\in [0,1/m]; \qquad\qquad F(y) = \frac{v-1/m}{v-y},\qquad y\in [0,1/m]. $$

Player $1$ bids ${b}_1=(x,0,...,0)$ and player $2$ bids
${b}_2=(y,...,y)$. $x$ and
$y$ are drawn from $G(x)$ and $F(y),$ respectively.  In case of a tie, the
item is always allocated to player $2$.

Let $\B=(B_1,B_2)$ denote this mixed bidding profile.  We are going to prove that $\B$ 
is a mixed Nash equilibrium for every $v>1/m.$

If player $1$ bids any $x$ in the range $(0,1/m]$ for  one item,
she gets the item with probability $F(x)$, since a tie occurs with
zero probability. Her expected utility is $F(x)(v-x)=v-1/m$. So, for
every $x\in (0,1/m]$ her utility is $v-1/m$. If player $1$ picks $x$
according to $G(x)$, her utility is still $v-1/m$, since she bids $0$
with zero probability. Bidding something greater than $1/m$ results in
a utility less than $v-1/m$. Regarding player $1$, it remains to show
that her utility when bidding for only one item is at least as high as her
utility when bidding for more items. Suppose player $1$ bids $(x_1,...,x_m)$,
where $x_i \geq x_{i+1}$, for $1\leq i \leq
m-1$. Player $1$ doesn't get any item if and
only if $y \geq x_1.$ So, with probability $F(x_1)$, she gets at least
one item and she pays at least $x_1$. Therefore, her expected utility
is at most $F(x_1)(v-x_1)=v-1/m$, but it would be strictly less if she
had nonzero payments for other items with positive probability. This
means that bidding only $x_1$ for one item and zero for the rest of
them dominates the strategy $(x_1,...,x_m).$

If player $2$ bids $y$ for all items, where $y \in[0,1/m]$, she gets
$m$ items with probability $G(y)$ and $m-1$ items with probability
$1-G(y)$. Her expected utility is
$G(y)(2-my)+(1-G(y))(1-(m-1)y)=G(y)(1-y)+1-(m-1)y=1$. Bidding
something greater than $1/m$ results in utility less than $1$. Suppose
now that player $2$ bids $(y_1,...,y_m)$, where $y_i \geq y_{i+1}$ 
for $1\leq i \leq m-1$. 
If $x \leq y_m$, player $2$ gets all the items; otherwise she
gets $m-1$ items and she pays her $m-1$ highest bids. So, her utility
is

\begin{eqnarray*}
  && G(y_m)\left(2-\sum_{i=1}^m y_i\right) + (1-G(y_m))\left(1-\sum_{i=1}^{m-1} y_i\right)\\
  &=& G(y_m)(1- y_m) + 1-\sum_{i=1}^{m-1} y_i\\
  &=& my_m + 1 - \sum_{i=1}^{m} y_i \\
  &\leq& my_m + 1 - \sum_{i=1}^{m} y_m = 1.
\end{eqnarray*}

Overall, we proved that $\B$ is a mixed Nash
equilibrium. It is easy to see that the social welfare in the optimum
allocation is $2$. In this Nash equilibrium, player $2$ bids $0$ with probability
$1-\frac{1}{mv}$, so, with at least this probability, player $1$ gets
one item.

\begin{eqnarray*}
  SW(\B) \leq \left(1-\frac{1}{mv} \right)(v+1)+\frac{1}{mv} 2 =1+v+\frac{1}{mv}-\frac{1}{m}
\end{eqnarray*}
If we set $v=1/\sqrt{m}$, then $SW(\B)\leq
1+\frac{2}{\sqrt{m}}-\frac{1}{m}$. So, $PoA \geq
\frac{2}{1+\frac{2}{\sqrt{m}}-\frac{1}{m}}$ which, for large $m$, converges to 2.
\end{proof}

%%% Local Variables: 
%%% mode: latex
%%% TeX-master: "poa"
%%% End: 

\paragraph{Acknowledgements} The authors are indebted to Orestis
Telelis for discussions on multi-unit auctions. We are also grateful to two anonymous reviewers for their useful comments and suggestions about a preliminary version of the paper.

%%% Local Variables: 
%%% mode: latex
%%% TeX-master: "poa"
%%% End: 

\bibliographystyle{plain}\bibliography{poa}
\begin{appendix}
  \section{Mixed Nash equilibria with submodular valuations}
\label{sec:LBd-dimensions}

Here, we construct a group of instances  with submodular
valuations, which have at least one mixed Nash equilibrium. 
These instances are a generalization of the lower bound of 
Section \ref{sec:submodular_LB}. We believe that this construction is interesting in its own right. 
\\

  Consider an instance with $n+1$ players and $n^d$
  items, where $d$ divides $n$.
 We will refer to the first $n$ players as the \emph{real} players 
and to the last one as the \emph{dummy} player.
 Let $[n]$ denote the set of
  integers $\{1,\ldots,n\}.$ We define the set of items as $M=[n]^d,$
  that is, they correspond to all the different vectors $w=(w_1, w_2,
  ... , w_d)$ with $w_i\in [n].$ Intuitively, they can be thought of
  as the nodes of a $d$ dimensional grid, with coordinates in $[n]$
  in each dimension.

We divide the \emph{real} players into $d$ groups of $n/d$ players.
 Let $g(i)$ denote the group that player $i$ belongs to. 
We associate each group with one of the
dimensions (directions) of the grid. In particular, for any fixed
player $i$, his valuation for a subset of items $S\subseteq M$ is the
size (number of elements) in the $d-1$-dimensional projection of $S$
in direction $g(i)$, times $v$. Formally, 
 $$v_i(S)=v|\{w_{-g(i)}\,|\,\exists w_{g(i)}\,\,s.t. \,\,(w_{g(i)},w_{-g(i)})\in S \}|.$$
 It is straightforward to check that
$v_i$ has decreasing marginal valuations, and is therefore submodular.
The valuation of the \emph{dummy} player for any subset of items $S\subseteq M$ is $(v-1)|S|$.

Given these valuations, we describe a mixed Nash equilibrium
$\B=(B_1,\ldots ,B_n)$ having a PoA $\left(1-\left(1-\frac{1}{n}\right)^{n+\frac{n}{d}-1}\right)^{-1}$ 
which is arbitrarily close to $\frac{1}{1-e^{-1-\frac{1}{d}}}$,
for large enough $n$. The dummy player bids $v-1$ for every item, and
receives the item if all of the real players bid at most $v-1$ for it. The
utility of the dummy player is always $0$. For real
players the mixed strategy $B_i$ is the following. Every player $i$
 picks a number $\ell\in [n]$ uniformly at random, and an $x$ according
to the distribution with CDF 
$$G(x)=(n-1)\left(\frac{1}{\left(v-x\right)^{\frac{1}{n-1}}}-1\right),$$
where $x\in\left[v-1,v-\left(\frac{n-1}{n}\right)^{n-1}\right]$.
Subsequently, he bids $x$ for every item $w=(\ell, w_{-g(i)})$, with
$w_{g(i)}=\ell$ as $g(i)^{th}$ coordinate, and bids $v-1$ for the rest of the
items. %That is, in any $b_i$ in the support of $B_i,$ the player bids a positive $x$ only for a $d-1$ dimensional slice of the items. Observe that $G(x)$ has no mass points, so tie-breaking matters
only in case of $v-1$ bids for an item, in which case the $dummy$ player
gets the item.

Let $F(x)$ denote the probability that any bidder $i$ gets a fixed item
$j$, given that he bids $b_i(j)=x>v-1$ for this item, and the bids in
$\b_{-i}$ are drawn from $\B_{-i}$ (due to symmetry, this probability
is the same for all items $w=(\ell, w_{-g(i)})$). For every other player
$k,$ the probability that he bids $v-1$ for item $j$ is $(n-1)/n,$ and
the probability that $j$ is in his selected slice but he bids lower
than $x$ is $G(x)/n.$ Multiplying over the $n-1$ other players, we
obtain
\[F(x)= \left(\frac{G(x)}{n}+\frac{n-1}{n}\right)^{n-1}=
\frac{\left(1-\frac{1}{n}\right)^{n-1}}{v-x}.\] Notice that $v_i$ is an
additive valuation restricted to the slice of items that player $i$
bids for in a particular $b_i.$ Therefore the expected utility of $i$
when he bids $x$ in $b_i$ is $F(x)(v-x)$ for one of these items,
and comprising all items
$\E[u_i(b_i)] =n^{d-1} F(x)(v-x)=n^{d-1} \left(1-1/n\right)^{n-1}$.

Next we show that for $v\geq
\left(1-\frac{1}{n}\right)^{-\frac{n}{d}+1}$, $\B$ is a Nash
equilibrium. In particular, the bids $b_i$ in the support of $B_i$
maximize the expected utility of a fixed player $i$.

First, we fix an arbitrary $w_{-g(i)},$ and focus on the set of items
$C:=\{(\ell,w_{-g(i)})\,|\, \ell \in [n]\},$ which we call a
\emph{column} for player $i$. Recall that $i$ is interested in getting
only one item within $C,$ on the other hand his valuation is additive
over items from different columns. Observe that restricted to a fixed
column, submitting any bid
$x\in\left[v-1,v-\left(\frac{n-1}{n}\right)^{n-1}\right]$ for one
arbitrary item results in the constant expected utility of $
\left(1-\frac{1}{n}\right)^{n-1}$, whereas a bid higher than
$v-\left(\frac{n-1}{n}\right)^{n-1}$ guarantees the item but pays more
so the utility becomes strictly less than $
\left(1-\frac{1}{n}\right)^{n-1}$ for this column.

We introduce two functions, $F_1(x)$ and $F_2(x)$,
\[F_1(x) = \left(\frac{G(x)}{n}+\frac{n-1}{n}\right)^{\frac{d-1}{d}n}
=\left(\frac{1-\frac{1}{n}}{(v-x)^\frac{1}{n-1}}\right)^{\frac{d-1}{d}n},\]
\[F_2(x) = \left(\frac{G(x)}{n}+\frac{n-1}{n}\right)^{\frac{n}{d}-1} =
\left(\frac{1-\frac{1}{n}}{(v-x)^\frac{1}{n-1}}\right)^{\frac{n}{d}-1}.\]
$F_1(x)$ denotes the probability that any bidder $i$ gets a fixed item
$j$, while the rest of the players in $g(i)$ bid $v-1$, given that
player $i$ bids $b_i(j)=x$ for this item, and the rest bids in
$\b_{-i}$ are drawn from $\B_{-i}$. Similarly, $F_2(x)$ denotes the
same probability while the players of all the other groups, apart from
$g(i)$, bid $v-1$. Notice that $F(x)=F_1(x)F_2(x)$. In any fixed
$\b_{-i}$, every other player of group $k\neq g(i)$ submits the same
bid for all items in $C,$ because either the whole $C$ is in the
current slice of $k$, and he bids the same value $x$, or no item from
the column is in the slice and he bids $v-1$. Therefore, $F_1(x)$ for
items in $C$ are \emph{fully dependent} distributions, whereas
$F_2(x)$ for items in $C$ are \emph{independent} distributions.

We first show that if bidder $i$ bids more than $v-1$ for at least two
items in $C$, bidding $v-1+\varepsilon$ for all of these items, for a
significantly small $\varepsilon>0$, is a better strategy. Suppose
that bidder $i$ bids $\mathbf{x}$, which is $x_j > v-1$ for every item
$j \in R \subseteq C$ and $v-1$ for the rest, where $k$ is the
cardinality of $R$. Reorder the items in $R$ in a way that the bids
are in non-increasing order. We will use mathematical induction over
$k$. Let $u_i(\mathbf{x},k)$ be the utility of player $i$ for bidding
$\mathbf{x}$, when the number of bids strictly greater than $v-1$ is
$k$.
\begin{eqnarray*}
E[u_i(\mathbf{x},2)] &=& F(x_2)F_2(x_1)(v-x_1-x_2) + F(x_2)(1-F_2(x_1))(v-x_2)\\
&&+F_2(x_1)(F_1(x_1)-F(x_2))(v-x_1)\\
&=& F(x_2)(v-x_2)+F(x_1)(v-x_1)-F(x_2)F_2(x_1)v\\
&=& 2\left(1-\frac{1}{n}\right)^{n-1}-F(x_2)F_2(x_1)v
\end{eqnarray*}
$E[u_i(\mathbf{x},2)]$ is maximized when both $x_1$ and $x_2$ are
minimized, so for $x_1=x_2=v-1+\varepsilon$. Let $R_{-1}$ denote the
set of all $k$ items apart from the first one (for which bidder $i$
bids the most).  Assume that $E[u_i(\mathbf{x},k-1)]$ is maximized
when $x_j=v-1+\varepsilon$, for all $j \in R_{-1}$. Moreover, let
$\P(S\neq \emptyset|S\subseteq R_{-1})$ and $\P(S)$ be the
probabilities that he gets at least one item from $R_{-1}$ and gets
$S$, respectively. It is easy to see that $\P(S\neq
\emptyset|S\subseteq R_{-1})=\sum_{\substack{S\subseteq R_{-1} \\S
    \neq \emptyset }}\P(S)$.
\begin{equation*}\begin{array}{ccl}
  E[u_i(\mathbf{x},k)] &=& F_2(x_1)(F_1(x_1)-\P(S\neq \emptyset|S\subseteq R_{-1}))(v-x_1) \\
  &&+ F_2(x_1)\sum_{S\subseteq R_{-1},S \neq \emptyset}\P(S)(v-\sum_{j \in
      S}x_j-x_1)\\
  &&+(1-F_2(x_1))E[u_i(\mathbf{x}_{-1},k-1)]    \\
  &=& F(x_1)(v-x_1)-F_2(x_1)\P(S\neq \emptyset|S\subseteq R_{-1})(v-x_1) \\
  &&+ F_2(x_1)E[u_i(\mathbf{x}_{-1},k-1)]-
  F_2(x_1)x_1\sum_{S\subseteq R_{-1},S \neq \emptyset
    }\P(S)\\
  &&+(1-F_2(x_1))E[u_i(\mathbf{x}_{-1},k-1)]\\
  &=& \left(1-\frac{1}{n}\right)^{n-1} + E[u_i(\mathbf{x}_{-1},k-1)] -F_2(x_1)\P(S\neq \emptyset|S\subseteq R_{-1})(v-x_1) \\
  &&- F_2(x_1)x_1\P(S\neq \emptyset|S\subseteq R_{-1})\\
  &=& \left(1-\frac{1}{n}\right)^{n-1} + E[u_i(\mathbf{x}_{-1},k-1)] -F_2(x_1)\P(S\neq \emptyset|S\subseteq R_{-1})v
\end{array}\end{equation*}
$F_2(x_1)$ is minimized when $x_1=v-1+\varepsilon$, for a
significantly small $\varepsilon>0$, $E[u_i(\mathbf{x}_{-1},k-1)]$ is
maximized when $x_j=v-1+\varepsilon$, for all $j \in R_{-1}$ and
$\P(S\neq \emptyset|S\subseteq R_{-1})$ is minimized for the same
values. So, $E[u_i(\mathbf{x},k)]$ is maximized when
$x_j=v-1+\varepsilon$, for all $j \in R$.

We next prove that for $v\geq
\left(1-\frac{1}{n}\right)^{-\frac{n}{d}+1}$, bidding $x>v-1$ only for
one item dominates the strategy of bidding $x$ for more than one
items.

\begin{lemma}\label{lem:binomMath}
  For any integer $k\geq 1$, $$\sum_{r=1}^k \binom{k}{r}x^r(1-x)^{k-r}
  = 1-(1-x)^k \text{ and }\sum_{r=1}^k \binom{k}{r}x^r(1-x)^{k-r}r =
  kx.$$
\end{lemma}

\begin{proof}
  \[\sum_{r=1}^k \binom{k}{r}x^r(1-x)^{k-r} = \sum_{r=0}^k
  \binom{k}{r}x^r(1-x)^{k-r} -
  \binom{k}{0}x^0(1-x)^{k-0}\]
  \[=(x+1-x)^k-(1-x)^k=1-(1-x)^k\]
  \begin{eqnarray*}
    \sum_{r=1}^k \binom{k}{r}x^r(1-x)^{k-r}r &=& \sum_{r=1}^k \frac{k!}{r!(k-r)!}x^r(1-x)^{k-r}r \\
    &=& k\sum_{r=1}^k \binom{k-1}{r-1}x^r(1-x)^{k-r}\\
    &=& kx\sum_{r=0}^{k-1} \binom{k-1}{r}x^r(1-x)^{k-1-r}\\
    &=& kx(x+1-x)^{k-1} = kx
  \end{eqnarray*}
\end{proof}

By using Lemma \ref{lem:binomMath}, the utility of bidding $x$ for $k$ items is,

\begin{eqnarray*}
E[u_i(x,k)] &=& F_1(x)\sum_{r=1}^k \binom{k}{r}F_2(x)^r(1-F_2(x))^{k-r}(v-rx)\\
&=& F_1(x)\left(\left(1-(1-F_2(x))^k\right)v - kF_2(x)x\right).
\end{eqnarray*}

We are going to bound the value of $v$, so that the utility decreases
as $k$ increases. So, we would like the following to hold.

\begin{equation*}
\begin{array}{ccl}
  E[u_i(x,k+1)] - E[u_i(x,k)] &\leq& 0 \\
  F_1(x)\left(\left(1-(1-F_2(x))^{k+1}\right)v - (k+1)F_2(x)x\right) \\
  -  F_1(x)\left(\left(1-(1-F_2(x))^k\right)v - kF_2(x)x\right) &\leq& 0 \\
  (1-F_2(x))^k v - (1-F_2(x))^{k+1}v - F_2(x)x  &\leq& 0\\
  F_2(x)(1-F_2(x))^k v - F_2(x)x  &\leq& 0\\
  (1-F_2(x))^k v -x  &\leq& 0
\end{array}
\end{equation*}

The quantity $(1-F_2(x))^k v -x$ is maximized when $x$ is minimized. Therefore,
\begin{eqnarray*}
(1-F_2(x))^k v -x &\leq& (1-F_2(v-1))^k v -v+1\\
&\leq& (1-F_2(v-1)) v -v+1\\
&=& 1 -F_2(v-1)v.
\end{eqnarray*}

Thus, $E[u_i(x,k+1)] - E[u_i(x,k)] \leq 0$, if $v\geq F_2^{-1}(v-1) =
\left(1-\frac{1}{n}\right)^{-\frac{n}{d}+1}$. Hence, for $v\geq
\left(1-\frac{1}{n}\right)^{-\frac{n}{d}+1}$, $\B$ is a Nash
equilibrium.

It remains to calculate the expected social welfare of $\B,$ and the
optimal social welfare.  We define a random variable w.r.t. the
distribution $\B.$ Let $Z_j=v$ if one of the real players $1,\ldots
,n$ gets item $j,$ and $Z_j=v-1$ if the \emph{dummy} player gets the
item. The expected social welfare is
\begin{eqnarray*}
\E_{\b\sim\B}[SW(\b)]&=&\sum_{j}E[Z_j]\\
&=&n^d((1-Pr(\text{\small no real player bids for  j}))v+ Pr(\text{\small no real player bids for  j})(v-1))\\
&=& n^d\left(v-\left(1-\frac{1}{n}\right)^{n}\right).
\end{eqnarray*}
Finally, we show that the optimum social welfare is $n^dv.$ An optimal
allocation can be constructed as follows: For each item $(w_1,w_2,
...,w_d)$ compute $r=\left(\sum_{i=1}^n w_i \mod n\right)$. Allocate
this item to the player $r$.  Similar to the Section
\ref{sec:submodular_LB}, each player is allocated $n^{d-1}$ items.

Therefore, the price of anarchy is
$\frac{v}{v-\left(1-\frac{1}{n}\right)^{n}}$.  For
$v=\left(1-\frac{1}{n}\right)^{-\frac{n}{d}+1}$, the price of anarchy
becomes
$\frac{\left(1-\frac{1}{n}\right)^{-\frac{n}{d}+1}}{\left(1-\frac{1}{n}\right)^{-\frac{n}{d}+1}-\left(1-\frac{1}{n}\right)^{n}}$
which for large $n$ it converges to $\frac{1}{1-e^{-1-\frac{1}{d}}}$.

It is easy to see that the case of $d=n$, is the special case of
Section \ref{sec:submodular_LB}, for which the Price of Anarchy is
$\frac{1}{\left(1-\left(1-\frac{1}{n}\right)^{n}\right)}$, and for
large $n$ converges to $\frac{1}{\left(1-\frac{1}{e}\right)}\approx
1.58$.

%%% Local Variables: 
%%% mode: latex
%%% TeX-master: "poa"
%%% End: 

  \section{A lower bound example for the single item Bayesian PoA.}
\label{sec:single_LB_Bayes}

Bayesian equilibria were
known to be inefficient~\cite{Kri02}. Here, for the sake of completeness, we present a lower bound example for the Bayesian price of
anarchy, with two players and only one item. 
\begin{theorem}\label{thm:addmixed1} For single-item auctions the PoA
  in Bayesian Nash equilibria is at least 1.06.
\end{theorem}
\begin{proof}
  In the lower-bound instance we have two bidders and only one
  item. The valuation of bidder~1 is always $1.$ Let $l=1-2/e,$ and
  $r=1-1/e.$ The valuation of bidder 2 is distributed according to the
  cumulative distribution function $H:$

  \begin{displaymath} H= \left\{ \begin{array}{ll}
        \frac{1}{e(1-l)} =\frac{1}{2},& x \in [0,l]\\
        \frac{1}{e(1-\frac{x+l}{2})}=\frac{2}{e(1-x)+2},&x \in [l,1].\\
      \end{array}\right.
  \end{displaymath}
  Observe that $v_2=0$ with probability $1/2,$ and $v_2$ is distributed over
  $[l,1]$ otherwise.  Consider the following bidding strategy
  $\B=(B_1,B_2):\,\,$ $B_1$ has a uniform distribution on $[l,r]$ with
  CDF $G(x)=\frac{x-l}{r-l}=ex-e+2$ on $[l,r];$ whereas the
  distribution $B_2$ is determined by the distribution of $v_2:$
 
  \begin{displaymath} b_2(v_2)= \left\{ \begin{array}{ll}
        0, & v_2=0\\
        \frac{v_2+l}{2},&v_2 \in [l,1].\\
      \end{array}\right.
  \end{displaymath}
  Let $F(x)$ denote the CDF of $b_2.$ We can compute the distribution
  of $b_2$ as follows. For $x\in [0,l),$ we have $F(x)=\P[b_2\le
  x]=\frac{1}{2};$ for $x\in [l,r]$ we have $F(x)=\P[b_2\le
  x]=\P[\frac{v_2+l}{2}\le x]=\P[v_2\le 2x-l]=\frac{1}{e(1-x)}.$
  Finally, $F(x)=1$ for $x\geq r.$ In summary, $b_2=0$ with
  probability $1/2,$ and is distributed over $[l,r]$ otherwise. On the
  other hand, $b_1$ is uniformly distributed over $[l,r].$ We do not
  need to bother about tie-breaking, since there are no mass points in
  $[l,r].$

  We prove next, that $\B$ is a Bayesian Nash equilibrium. Consider first player
  $2.$ If $v_2=0,$ his utility is clearly maximized. If $v_2\in
  [l,1],$ then $\E[u_2(b_2)]=G(b_2)(v_2-b_2).$ By straightforward
  calculation we obtain that over $[l,r]$ the function
  $G(x)(v_2-x)=\frac{(x-l)(v_2-x)}{r-l}$ is maximized in
  $b_2=(v_2+l)/2,$ so $b_2(v_2)$ is best response for bidder $2.$

  Consider now player 1. Given that over $[l,r]$ the distribution of
  $b_2$ is $F(x)=\frac{1}{e(1-x)},$ every bid $b_1\in [l,r]$ is best
  response for player 1, since his utility $F(b_1)(1-b_1)=1/e$ is
  constant.  Now we are ready to compute the social welfare of this
  Nash equilibrium.

  \[SW(\B)=Pr[v_2\le l]\cdot 1+\int_l^1(v_2\cdot
  G(b_2(v_2))+1-G(b_2(v_2)))\cdot h(v_2)d v_2\leq 0.942\] So the PoA
  is at least $1.06$.
\end{proof}

%%% Local Variables: 
%%% mode: latex
%%% TeX-master: "poa"
%%% End: 

  \section{Anonymity assumption}
\label{sec:anonymity}

Recall that we assume that $q_j^w(x)$ (or $q_j^l(x)$)
is the same for all bidders. 
Without the anonymity of $q^w_j(\cdot)$ over buyers, we can show by the following example that the
PoA is unbounded. Suppose there is a single
item to be sold to two players with valuation $v_1=1$ and
$v_2=\epsilon$. The losing payment is $0$ for both players but the
winning payments are different such that $q^w(x)=x$ for bidder $1$ and
$\bar{q}^w(x)= \epsilon \cdot x$ for bidder $2$. If there is a
tie, then the item is allocated to player $2.$ Now consider the bidding strategy $b_1=b_2=1$. It is easy to see
that it forms a Nash Equilibrium and has $PoA = 1/\epsilon$.

%%% Local Variables: 
%%% mode: latex
%%% TeX-master: "poa"
%%% End: 

\end{appendix}

\end{document}